\documentclass[journal]{IEEEtran}

\ifCLASSINFOpdf
\else
   \usepackage[dvips]{graphicx}
\fi
\usepackage{url}

\usepackage{cite}
\usepackage{amsmath,amssymb,amsfonts}
\usepackage{algorithmic}
\usepackage{graphicx}
\usepackage{textcomp}

\usepackage{color}
\usepackage{soul}
\usepackage[]{graphicx}
\usepackage{amsthm}
\usepackage{amsmath}
\usepackage[varg]{txfonts}
\usepackage{bm}
\usepackage{subfigure}
\usepackage{subfigmat}
\usepackage{algorithm}
\usepackage{algorithmic}
\usepackage{array}
\usepackage{stfloats}
\usepackage{breqn}
\usepackage{threeparttable}
\usepackage{epstopdf}
\newcommand{\argmax}{\mathop{\rm{arg~max}}\limits}
\newcommand{\argmin}{\mathop{\rm{arg~min}}\limits}

\newtheorem{proposition}{Proposition}
\newtheorem{definition}{Definition}[section]
\usepackage[utf8]{inputenc}
\usepackage[T1]{fontenc}
\usepackage[english]{babel}
\usepackage{tabularx}

\begin{document}

\title{Effect of Objective Function on Data-Driven Greedy Sparse Sensor Optimization}
\author{
\uppercase{Kumi Nakai},
\uppercase{Keigo Yamada},
\uppercase{Takayuki Nagata},
\uppercase{Yuji Saito}, and
\uppercase{Taku Nonomura}
\thanks{
K. Nakai is with Tohoku University, Sendai, 980-8579, Japan (e-mail: kumi.nakai@tohoku.ac.jp),
K. Yamada is with Tohoku University, Sendai, 980-8579, Japan,
T. Nagata is with Tohoku University, Sendai, 980-8579, Japan,
Y. Saito is with Tohoku University, Sendai, 980-8579, Japan,
T. Nonomura is with Tohoku University, Sendai, 980-8579, Japan.
}
}


\markboth{Journal of \LaTeX\ Class Files, Vol. 14, No. 8, August 2015}
{Shell \MakeLowercase{\textit{et al.}}: Bare Demo of IEEEtran.cls for IEEE Journals}
\maketitle
\begin{abstract}
The problem of selecting an optimal set of sensors estimating a high-dimensional data is considered. Objective functions based on D-, A-, and E-optimality criteria of optimal design are adopted to greedy methods, that maximize the determinant, minimize the trace of the inverse, and maximize the minimum eigenvalue of the Fisher information matrix, respectively. First, the Fisher information matrix is derived depending on the numbers of latent state variables and sensors. Then, a unified formulation of the objective function based on A-optimality is introduced and proved to be submodular, which provides the lower bound on the performance of the greedy method. Next, greedy methods based on D-, A-, and E-optimality are applied to randomly generated systems and a practical dataset concerning the global climate; these correspond to an almost ideal and a practical case in terms of statistics, respectively. The D- and A-optimality-based greedy methods select better sensors. The E-optimality-based greedy method does not select better sensors in terms of the index of E-optimality in the oversample case, while the A-optimality-based greedy method unexpectedly does so in terms of the index of E-optimality. The poor performance of the E-optimality-based greedy method is due to the lack of submodularity in the E-optimality index and the better performance of the A-optimality-based greedy method is due to the relation between A- and E-optimality. Indices of D- and A-optimality seem to be important in the ideal case where the statistics for the system are well known, and therefore, the D- and A-optimality-based greedy methods are suitable for accurate reconstruction. On the other hand, the index of E-optimality seems to be critical in the more practical case where the statistics for the system are not well known, and therefore, the A-optimality-based greedy method performs best because of its superiority in terms of the index of E-optimality. 
\end{abstract}
\maketitle
\section{Introduction}
\label{sec:introduction}
\IEEEPARstart{T}{he} development of an accurate and efficient model for estimation, prediction, and control of complex phenomena is an open challenge in various scientific and industrial domains. By virtue of innovations in measurement equipment and technology, it is possible to obtain vast amounts of data, such as seismic data concerning earthquake phenomena and environmental data from remote platforms using satellites. However, these phenomena may involve multidimensional states over various timescales. It is impractical to process the full-state measurements in real time because this is computationally expensive and therefore not conducive to fast state estimation for low-latency and high-bandwidth control.

To navigate this issue, dimensionality reduction is a promising approach.  There are often a few dominant low-dimensional patterns, which may well explain the high-dimensional data, in many natural science systems. Singular value decomposition (SVD) provides a systematic way to determine a low-dimensional approximation of high-dimensional data based on proper orthogonal decomposition (POD)\cite{berkooz1993proper,taira2017modal}. 
Thus, SVD enables us to exploit the significant POD modes in the data for low-dimensional representations. 

Sparse sensing is also important because the number of sensors is often limited because of the cost associated with their placement and computational constraints. Low-dimensionally approximated full states can be reconstructed from a small subset of measurements by sparse sensors.  
Thus, it is important to optimize sensor placement to exploit significant low-dimensional patterns based on efficient reduced-order models. This idea was adopted by Manohar et al.\cite{manohar2018data}, and a sparse-sensor-placement algorithm has been developed and discussed\cite{manohar2019optimized,manohar2018optimal}.
Furthermore, sensor selection based on POD is a data-driven approach without the requirement for governing equations. Such data-driven sensing generally needs to determine the optimal sensor locations from a large amount of candidates. Hence, a fast greedy optimization method is required for high-performance computing or feedback control.

\begin{center}
\begin{table*}[bht]
    \small
    \centering
    \caption{Summary of proposed optimization methods for sensor selection} 
    \label{table:Summary_optimization_method}
    \begin{tabular}{l|p{4cm}|p{8cm}}
        \hline
        Optimality criteria & Objective function & Scheme\\
        \hline\hline
        D-optimality & Maximization of the determinant of the Fisher information matrix & \begin{tabular}{p{8cm}}
            \begin{itemize}
                \item Convex relaxation method\cite{joshi2009sensor,nonomura2021randomized} 
                \item Greedy method 
                \begin{itemize}
                     \item QR-based greedy method\cite{manohar2018data,manohar2019optimized,manohar2018optimal}
                    \item Hybrid greedy method of QR and a straightforward D-optimality-based method\cite{saito2019determinant}
                    \item Noise-robust greedy method\cite{yamada2021fast}
                \end{itemize}
            \end{itemize}
        \end{tabular} \\
        \hline
        A-optimality & Minimization of the trace of Fisher information matrix &
        \begin{tabular}{p{8cm}} 
            \begin{itemize}
                \item Greedy method\cite{krause2008near,nguyen2018efficient}
                \item Convex relaxation method\cite{nagata2020data}
                \item Proximal splitting algorithm-based method\cite{nagata2020data}
            \end{itemize}
        \end{tabular} \\
        \hline
        E-optimality & Maximization of the minimum eigenvalue of the Fisher information matrix & 
        \begin{tabular}{p{8cm}} 
            \begin{itemize}
                \item Greedy method\cite{nguyen2018efficient}
                \item Quasi-greedy method using a lower bound\cite{krause2008near}
            \end{itemize}
        \end{tabular} \\
        \hline
    \end{tabular}
\end{table*}
\end{center}

The optimal sensor selection problem is closely related to the optimal design of experiments, which provides small values of the variances of estimated parameters and predicted response\cite{Atkinson2007optimum}. Depending on the statistical criterion of optimal experimental design, the objective function of sensor selection problems is defined using the Fisher information matrix, which corresponds to the inverse of the covariance matrix of an estimator. Table~\ref{table:Summary_optimization_method} shows various optimization methods typically used for the sensor selection problem from the perspective of optimal design.
The most important design criterion is that of D-optimality, in which the determinant of the Fisher information matrix is maximized. Here, `D' stands for `determinant.' 
This criterion results in minimization of the volume of the confidence ellipsoid of the regression estimates.
Joshi and Boyd adopted D-optimality in a sensor selection problem. They proposed to utilize a convex relaxation method and solved the approximate problem of D-optimal design\cite{joshi2009sensor}. This algorithm was recently improved with the development of a randomized algorithm\cite{nonomura2021randomized}. Manohar et al.\cite{manohar2018data} proposed a greedy method based on the discrete-empirical-interpolation method (DEIM) and QR-DEIM (QDEIM) \cite{drmac2016new,chaturantabut2010nonlinear}; these are methods in the framework of reduced-order modeling using sparse sampling points. Manohar et al. showed that this method was advantageous in terms of being significantly faster than the convex optimization method\cite{manohar2018data}. Their greedy method optimizes the sensor location by QR-pivoting the row vector of the sensor-candidate or related matrix. This greedy method was shown\cite{saito2020data} to correspond to selecting the row vector the norm of which is the maximum and eliminating its component from the rest of the matrix using a Gram-Schmidt procedure, though QR implementation is much faster in practical contexts. More recently, Saito et al.\cite{saito2019determinant} mathematically illustrated that the objective function adopted by Manohar et al.\cite{manohar2018data} corresponds to the maximization of the D-optimality objective function when the number of sensors is less than the number of state variables, and derived a unified expression of the objective function based on D-optimality regardless of the number of sensors. Furthermore, they successfully proposed an efficient greedy method based on D-optimality: a hybrid of greedy methods based on QR decomposition and the straightforward maximization of the determinant in the case where the number of sensors is less and greater than that of state variables, respectively. The proposed method is confirmed to provide nearly optimal sensors as well as the convex approximation method but significantly reduces the computational cost compared to the convex approximation method and QR-based greedy method\cite{saito2019determinant}. In addition, a greedy method for sensor selection problems under correlated noise has also been developed in the framework of extended D-optimality.\cite{yamada2021fast} Thus, previous studies have provided valuable knowledge on the sensor selection problem in the context of D-optimality. 

There are a variety of criteria for optimal design. Two other criteria that have a statistical interpretation in terms of the information matrix are A- and E-optimality. In A-optimality, the trace of the inverse of the information matrix, which corresponds to the total variance of the parameter estimates, is minimized. It is equivalent to minimization of the `average' variance, hence `A' stands for `averaged.' In E-optimality, the maximum eigenvalue of the information matrix is minimized, which minimizes the worst-case variance of estimation error, where `E' stands for `eigenvalue.'
Ojective functions based on these criteria have been introduced in\cite{Atkinson2007optimum,joshi2009sensor}, and greedy methods based on D-, A-, and E-optimality have been adopted in combinatorial optimization problems\cite{krause2008near,nguyen2018efficient}. 
Krause et al.\cite{krause2008near} compared placements using greedy methods based on D-, A-, and E-optimality with those using the proposed method, which maximizes the mutual information between the chosen sensors and the candidate sensor for the Gaussian process regression. The results of numerical experiments show that the proposed method tends to outperform the classical D-, A-, and E-optimal design, and the performance of the D-, A-, and E-optimality-based greedy methods varies depending on the problem. Moreover, Nguyen et al.\cite{nguyen2018efficient} extended their method in the context of A-optimality, and proposed a separable optimization approach, which only depends on spatial variations in spatio-temporal environments, in order to reduce the computational complexity of the optimization problem. However, in the previous studies\cite{krause2008near,nguyen2018efficient}, the performance characteristics of the greedy methods based on D-, A-, and E-optimality were not discussed in detail. Furthermore, since the latent state variables are equipped with prior information (averaged value and variance), and the problem is regularized from the selection of the first sensor, the formulation for the greedy sensor selection for the underdetermined situation has not been derived. Note that the formulation based on D-optimality has been discussed in Saito et al.\cite{saito2019determinant}. Thus, to the best of our knowledge, the greedy method and its performance for A- and E-optimality in underdetermined situations has not yet been investigated.

For linear dynamical systems, objective functions based on D-, A-, and E-optimal design were adopted in Summers et al.\cite{summers2015submodularity}. The determinant, trace, and minimum eigenvalue of the controllability Gramian, which corresponds to the Fisher information matrix for non-dynamical systems, are considered. Similar to linear dynamical systems, it is also important to understand the characteristics of objective functions for the reconstruction of snapshots of systems. This would provide more fundamental knowledge for the sparse sensor selection problem.

The aim of this study is to obtain insights into the objective functions suitable for the sparse sensor selection problem, especially focusing on the reconstruction of snapshots of high-dimensional data by the greedy method. For this purpose, greedy methods based on D-, A-, and E-optimality including the underdetermined situation are proposed and described, and evaluated in terms of their performance.  
Firstly, the submodularity of objective functions is investigated. Next, the performance of the greedy methods based on D-, A-, and E-optimality is evaluated with respect to two problems with ideal and nonideal conditions in terms of statistical assumptions. The indices of optimality criteria, reconstruction error, and computational cost are compared. 
The main contributions of the paper are as follows:
\begin{itemize}
    \item The Fisher information matrix for underdetermined cases is derived in the observable subspace of the measurement matrix in Section \ref{sec:formulation}.
    \item Objective functions based on A- and E-optimality for both underdetermined and overdetermined situations are derived in Section \ref{sec:formulationA} and \ref{sec:formulationE}. Greedy methods based on A- and E- optimality in the underdetermined situation without regularization terms are introduced for the first time; in previous studies, they have only been discussed for the overdetermined situation.
    \item A unified formulation of the objective function based on A-optimality is introduced and proved to be submodular, which provides the lower bound on the performance of the greedy method in Section \ref{sec:submodularityA}. On the other hand, the objective function based on E-optimality is proved to be neither submodular nor supermodular while it is monotone in the overdetermined case in Section \ref{sec:submodularityE}.
    \item Greedy methods based on D-, A-and E-optimality are adopted with respect to ideal and nonideal problems in terms of the statistics, and guidance on selecting a suitable objective function for a given dataset and situation is provided in Section \ref{sec:results}. 
\end{itemize}

\section{Formulation and Algorithms for Sensor Selection Problem}
\label{sec:formulation}
We consider the linear system given by
\begin{align}
\bm{y}&=\bm{HUz}=\bm{C}\bm{z}\label{eq:y_cz},
\end{align}
where $\bm{y} \in \mathbb{R}^p$, $\bm{H} \in \mathbb{R}^{p\times n}$, $\bm{U} \in \mathbb{R}^{n\times r}$, $\bm{z} \in \mathbb{R}^r$, and $\bm{C} \in \mathbb{R}^{p\times r}$ are the observation vector, the sparse sensor location matrix, and the sensor candidate matrix, the latent state vector, and the measurement matrix ($\bm{C}=\bm{HU}$), respectively. Here, the element corresponding to the sensor location is unity and the others are 0 in each row of $\bm{H}$. In addition, $p$, $n$, and $r$ are the number of sensors, the number of spatial dimensions, and the number of latent state variables, respectively. The system above represents the problem of choosing $p$ observations out of $n$ sensor candidates for the estimation of the state variables. The various sensor selections can be expressed by changing $\bm{H}$ and by selecting row vectors as sensors from the sensor candidate matrix $\bm{U}$.   
A graphical image of the foregoing equation is shown in Fig.~\ref{fig:Graphicalimag}. 

\begin{figure}[htbp]
    \centering
    \includegraphics[width=3in]{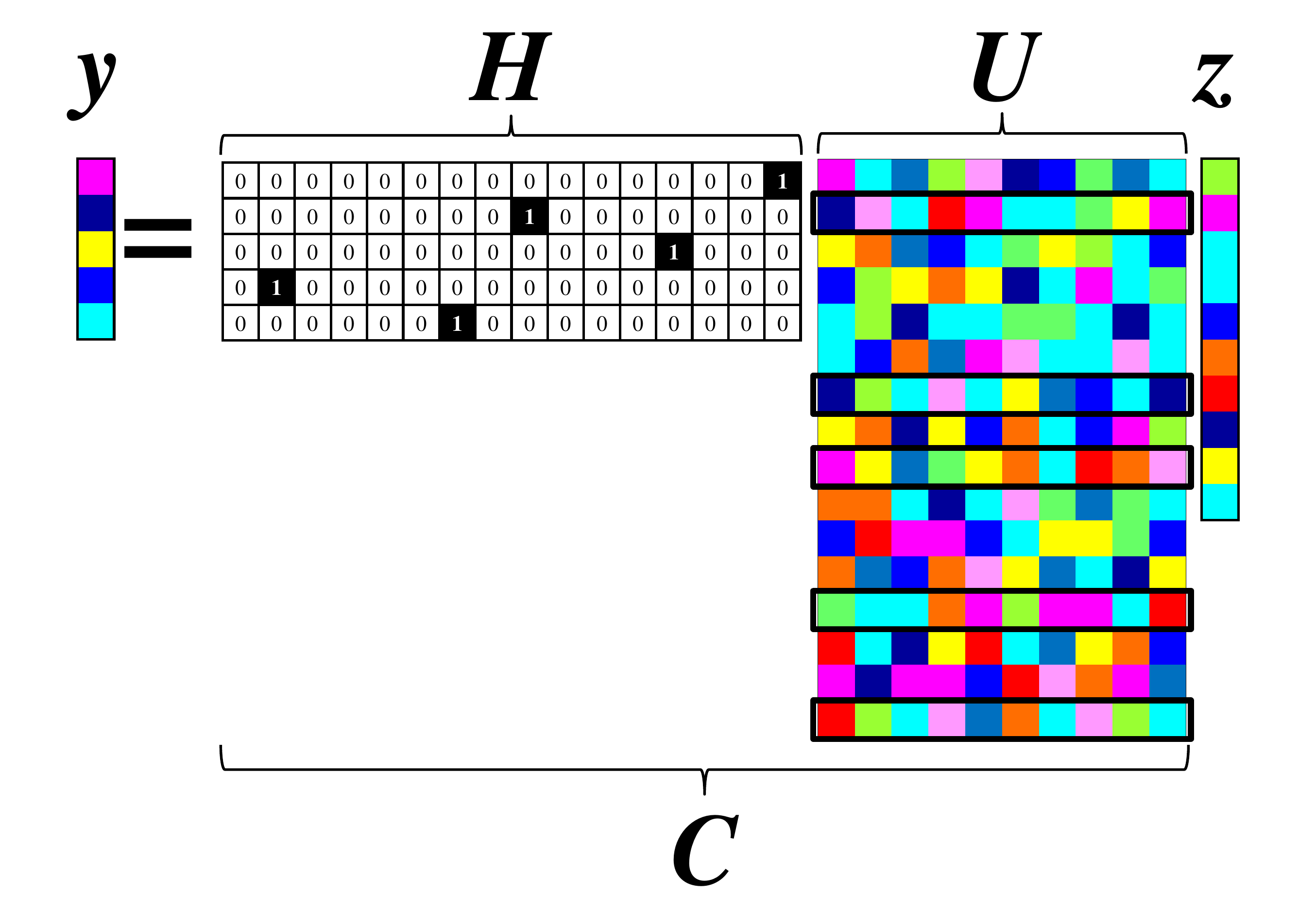}
    \caption{Graphical image for sensor matrix $\bm{H}$ on \eqref{eq:y_cz}\cite{saito2019determinant}}
    \label{fig:Graphicalimag}
\end{figure}

The estimated parameters $\bm{\hat{z}}$ can be obtained by the pseudo-inverse operation when uniform independent Gaussian noise $\mathcal{N}(\bm{0},\sigma^2\bm{I})$ is imposed on the observations as follows:
\begin{align}
    \bm{\hat{z}}
    &=\bm{C}^{+}\bm{y}
    =\left\{\begin{array}{cc}\bm{C}^{\top}\left(\bm{C}\bm{C}^{\top}\right)^{-1}\bm{y}, & p\le r \\
    \left(\bm{C}^{\top}\bm{C}\right)^{-1}\bm{C}^{\top}\bm{y}, & p>r \end{array} 
    \right.\label{eq:LS_estimation}.
\end{align}

The covariance matrix of estimation error is expressed as follows:
\begin{align}
    &{E}\left[\left(\bm{z}-\bm{\hat{z}}\right)
        \left(\bm{z}-\bm{\hat{z}}\right)^{\top}\right] \notag\\
    &=\left\{\begin{array}{cc}
        (\bm{I}-\bm{P}_{\bm{C}})E[\bm{z}\bm{z}^{\top}](\bm{I}-\bm{P}_{\bm{C}})+\sigma^2\bm{C}^{\top}\left(\bm{C}\bm{C}^{\top}\right)^{-2}\bm{C}, & p\le r, \\
        \sigma^2\left(\bm{C}^{\top}\bm{C}\right)^{-1}, & p>r, 
    \end{array}\right.\label{eq:error_covarz}
\end{align}
where $\bm{P}_{\bm{C}}=\bm{C}^{\top}(\bm{C}\bm{C}^{\top})^{-1}\bm{C}$ is the projection matrix onto the row vector space of $\bm{C}$.
Here, the estimation error in the case where $p<r$ is further considered in the observable subspace by the transformation. The full singular value decomposition of $\bm{C}$ is given as follows:
\begin{align}
\bm{C}&=\bm{U}_{\bm{C}}\bm{\Sigma}_{\bm{C}}\bm{V}_{\bm{C}}^{\top} \notag \\
&=
\left\{\begin{array}{cc}
\bm{U}_{\bm{C}} 
\left[\begin{array}{cc}\tilde{\bm{\Sigma}_{\bm{C}}} & \bm{0} \end{array}\right]
\left[\begin{array}{c}\tilde{\bm{V}_{\bm{C}}}^{\top} \\ \bar{\bm{V}_{\bm{C}}}^{\top} \end{array}\right], & p \le r,\\
\bm{U}_{\bm{C}}\left[\begin{array}{c}\tilde{\bm{\Sigma}_{\bm{C}}}\\\bm{0}\end{array}\right] \tilde{\bm{V}_{\bm{C}}}^{\top}, & p > r. 
\end{array}\right.
\end{align}
The coordinate transform by $\bm{\zeta}=\tilde{\bm{V}_{\bm{C}}}^{\top}\bm{z}$ is considered, whereas $\bm{V}_{\bm{C}}$ is an orthonormal (or unitary) matrix and the amplitude of the error in the observable subspace does not change. This transformation only gives the observable components in the latent state vector. The error in $\bm{\zeta}$ becomes as follows:
\begin{align}
\bm{E}\left[\left(\bm{\zeta}-\bm{\hat{\zeta}}\right)
    \left(\bm{\zeta}-\bm{\hat{\zeta}}\right)^{\top}\right] 
    =\left\{ \begin{array}{cc}
    \sigma^2\bm{U}_{\bm{C}}^{\top} \left(  \bm{C}\bm{C}^{\top} \right)^{-1} \bm{U}_{\bm{C}}, & p \le r, \\
    \sigma^2\tilde{\bm{V}_{\bm{C}}}^{\top} \left(\bm{C}^{\top}\bm{C}\right)^{-1}\tilde{\bm{V}_{\bm{C}}}, & p > r. 
    \end{array} \right. \label{eq:error_covarzeta}
\end{align}
Here, the only error covariance in the observable space is obtained in the case of $p \le r$ and the amplitude of the error covariance matrix is given by $(\bm{CC}^{\top})^{-1}$ because row vectors in $\tilde{\bm{V}_{\bm{C}}}^{\top}$ are orthonormal to each other and their absolute value is unity. Here, the first term of (\ref{eq:error_covarz}) in the case of $p<r$ disappears in (\ref{eq:error_covarzeta}) because the term only has components in the unobservable space. 
In the optimal design, the error covariance matrix and its inverse, the latter of which corresponds to the Fisher information matrix, are employed, and the optimality criteria are provided. 
Sensor selection problems can be defined based on the optimality criterion. There are a variety of optimality criteria, several of which are addressed hereafter. 

\subsection{Objective function based on D-optimality}
\label{sec:formulationD}
A D-optimal design maximizes the determinant of the Fisher information matrix. It is equivalent to minimizing the determinant of the error covariance matrix, resulting in minimizing the volume of the confidence ellipsoid of the regression estimates of the linear model parameters. 
Therefore, the problem can be expressed as the optimization problem
\begin{align}
    &\mathrm{maximize}\,\,f_{\mathrm{D}} \nonumber \\
    &f_{\mathrm{D}}\,=\,\left\{\begin{array}{cc}
        \mathrm{det}\,\left(\bm{C}\bm{C}^{\top}\right), & p\le r, \\
        \mathrm{det}\,\left(\bm{C}^{\top}\bm{C}\right), & p>r.
    \end{array}\right.\label{eq:obj_det}
\end{align}
All the combinations of $p$ sensors out of $n$ sensor candidates should be searched by a brute-force algorithm for the real-optimized solution of (\ref{eq:obj_det}), which takes an enormous amount of computational time ($O(n!/(n-p)!/p!)\approx O(n^{p})$).
Instead, greedy methods for the suboptimized solution have been devised by adding a sensor step by step. 
For the D-optimal criterion, the objective function for the greedy method has already been demonstrated in the literature\cite{saito2019determinant}.

\subsection{Objective function based on A-optimality}
\label{sec:formulationA}
An A-optimal design minimizes the mean square error in estimating the parameter $z$. Hence, the objective function is the trace of the error covariance matrix in A-optimal design. The sensor selection problem can be expressed as the following optimization problem:
\begin{align}
    &\mathrm{minimize}\,\,f_{\mathrm{A}} \nonumber \\
    &f_{\mathrm{A}}\,=\,\left\{\begin{array}{cc}
        \mathrm{tr}\,\left[\left(\bm{C}\bm{C}^{\top}\right)^{-1}\right], & p\le r, \\
        \mathrm{tr}\,\left[\left(\bm{C}^{\top}\bm{C}\right)^{-1}\right], & p>r.
    \end{array}\right.\label{eq:obj_tr}
\end{align}
In the step-wise selection of the greedy method, selection of only the $k$th sensor is carried out in the $k$th step under the condition that the sensors up to $(k-1$)th are already determined. 
Let $S \subset \{ 1, 2, \dots, n \}$ be a set of labels of selected sensors and $\bm{C}_{S}$ be the corresponding sensor matrix.
Specifically, if ${S} = \{ i_{1}, i_{2}, \dots, i_{k} \}$, then $\bm{C}_{S}$ is given by
\begin{align}
\bm{C}_{{S}}=
    \bm{C}_{k} 
    =\left[ 
        \begin{array}{cccc}
            \bm{u}_{i_{1}}^{\top} &
            \bm{u}_{i_{2}}^{\top} & \cdots &
            \bm{u}_{i_{k}}^{\top}
        \end{array}
    \right]^{\top},
\end{align}
where $\bm{u}_{i_{k}}$ is the corresponding row vector of the sensor-candidate matrix $\bm{U}$. 
Therefore, the sensor index chosen in the $k$th step of the greedy method can be described as follows:
\begin{align}
    i_{k}=&\argmin_{i_k}\,\,f_{\mathrm{AG}}, \nonumber 
\end{align} where
\begin{align}
    &f_{\mathrm{AG}}\,=\,\left\{\begin{array}{cc}
        \mathrm{tr}\,\left[\left(\bm{C}_k\bm{C}_k^{\top}\right)^{-1}\right], & p\le r, \\
        \mathrm{tr}\,\left[\left(\bm{C}_k^{\top}\bm{C}_k\right)^{-1}\right], & p>r.
    \end{array}\right. \label{eq:obj_tr2}
\end{align}
In the case of $p \le r$, 
\begin{align}
    &\mathrm{tr}\left[\left(\bm{C}_{k}\bm{C}_{k}^{\top}\right)^{-1}\right] \nonumber \\
    =\,&\mathrm{tr}\left[ \left(\left[
        \begin{array}{cc} \bm{C}_{k-1} \\ \bm{u}_{i_k} \end{array}
        \right] \left[
        \begin{array}{cc} \bm{C}_{k-1}^{\top} &  \bm{u}_{i_k}^{\top} \end{array}
        \right]\right)^{-1} \right] \nonumber \\
    =\,&\mathrm{tr}\left[\left(\bm{C}_{k-1}\bm{C}_{k-1}^{\top}\right)^{-1}\right] \nonumber \\
    &+\mathrm{tr}\left[\frac
        {\left(\bm{C}_{k-1}\bm{C}_{k-1}^{\top}\right)^{-1}
        \bm{C}_{k-1}\bm{u}_{i_k}^{\top}\bm{u}_{i_k}\bm{C}_{k-1}^{\top}
        \left( \bm{C}_{k-1}\bm{C}_{k-1}^{\top} \right) ^{-1}}
        {\bm{u}_{i_k} \left( \bm{I}-\bm{C}_{k-1}^{\top}
        \left(\bm{C}_{k-1}\bm{C}_{k-1}^{\top}\right)^{-1}
        \bm{C}_{k-1} \right) \bm{u}_{i_k}^{\top} }\right] \nonumber \\
    &+\mathrm{tr}\left[\frac{1}
        {\bm{u}_{i_k}\left(\bm{I}-\bm{C}_{k-1}^{\top}
        \left(\bm{C}_{k-1}\bm{C}_{k-1}^{\top}\right)^{-1}
        \bm{C}_{k-1}\right)\bm{u}_{i_k}^{\top}}\right].
    \label{eq:obj_tr_greedy_r}
\end{align}
Here, it is not necessary to evaluate the first term of (\ref{eq:obj_tr_greedy_r}) of the last equation since $\left(\bm{C}_{k-1}\bm{C}_{k-1}^{\top}\right)^{-1}$ is already determined in the $k$th step. Considering the cyclic property of trace, the greedy methods can be simply written as follows:
\begin{align}
    i_{k}=\argmin_{i_k}\,\,\frac
        {\bm{u}_{i_k}\bm{C}_{k-1}^{\top}\left(\bm{C}_{k-1}\bm{C}_{k-1}^{\top}\right) ^{-2}\;
        \bm{C}_{k-1}\bm{u}_{i_k}^{\top} + 1}
        {\bm{u}_{i_k}\left(\bm{I}-\bm{C}_{k-1}^{\top}\left(\bm{C}_{k-1}\bm{C}_{k-1}^{\top}\right)^{-1}\bm{C}_{k-1}\right)\bm{u}_{i_k}^{\top}},\,p\le r.
\end{align}
In the case of $p>r$,
\begin{align}
    &\mathrm{tr}\left[\left(\bm{C}_{k}^{\top}\bm{C}_{k}\right)^{-1}\right] \nonumber \\
    =\,&\mathrm{tr}\left[\left(\bm{C}_{k-1}^{\top}\bm{C}_{k-1}+\bm{u}_{i}^{\top}\bm{u}_{i}\right)^{-1} \right] \nonumber \\
    =\,&\mathrm{tr}\left[\left(\bm{C}_{k-1}^{\top}\bm{C}_{k-1}\right)^{-1}\right] \nonumber \\
     &-\mathrm{tr}\left[\left(\bm{C}_{k-1}^{\top}\bm{C}_{k-1}\right)^{-1}
     \bm{u}_{i}^{\top}\left(1+\bm{u}_{i}\left(\bm{C}_{k-1}^{\top}\bm{C}_{k-1}\right)^{-1}\bm{u}_{i}^{\top}\right)^{-1}\right. \nonumber \\
     &\qquad\left.\bm{u}_{i}\left(\bm{C}_{k-1}^{\top}\bm{C}_{k-1}\right)^{-1}\right].
    \label{eq:obj_tr_greedy_p}
\end{align}
Taking into account the fact that the first term of (\ref{eq:obj_tr_greedy_p}) does not contribute in the $k$th sensor selection, the greedy method can again be simply written as follows:
\begin{align}
    i_{k}=\argmin_{i_k}\,\,-\frac
        {\bm{u}_{i}
        \left(\bm{C}_{k-1}^{\top}\bm{C}_{k-1}\right)^{-2}
        \bm{u}_{i}^{\top}}
        {1+\bm{u}_{i}\left(\bm{C}_{k-1}^{\top}\bm{C}_{k-1}\right)^{-1}\bm{u}_{i}^{\top}},\,p>r.
\end{align}
In summary, the greedy method can be written as follows:
\begin{align}
    i_{k}=\left\{\begin{array}{ll}
        \argmin_{i_k}\,\,\frac
        {\bm{u}_{i_k}\bm{C}_{k-1}^{\top}\left(\bm{C}_{k-1}\bm{C}_{k-1}^{\top}\right) ^{-2}\;
        \bm{C}_{k-1}\bm{u}_{i_k}^{\top} + 1}
        {\bm{u}_{i_k}\left(\bm{I}-\bm{C}_{k-1}^{\top}\left(\bm{C}_{k-1}\bm{C}_{k-1}^{\top}\right)^{-1}\bm{C}_{k-1}\right)\bm{u}_{i_k}^{\top}}, & p\le r, \\
        \argmin_{i_k}\,\,-\frac
        {\bm{u}_{i}
        \left(\bm{C}_{k-1}^{\top}\bm{C}_{k-1}\right)^{-2}
        \bm{u}_{i}^{\top}}
        {1+\bm{u}_{i}\left(\bm{C}_{k-1}^{\top}\bm{C}_{k-1}\right)^{-1}\bm{u}_{i}^{\top}}, & p>r.
    \end{array}\right. \label{eq:obj_tr_greedy}
\end{align}

\subsection{Objective function based on E-optimality}
\label{sec:formulationE}
An E-optimal design minimizes the worst-case variance of estimation error, which corresponds to the maximum eigenvalue of the Fisher information matrix. Therefore, the sensor selection problem can be expressed as the following optimization problem:
\begin{align}
    &\mathrm{maximize}\,\,f_{\mathrm{E}} \nonumber \\
    &f_{\mathrm{E}}\,=\,\left\{\begin{array}{cc}
        \lambda_{\mathrm{min}}\,\left(\bm{C}\bm{C}^{\top}\right), & p\le r, \\
        \lambda_{\mathrm{min}}\,\left(\bm{C}^{\top}\bm{C}\right), & p>r.
    \end{array}\right.\label{eq:obj_eig}
\end{align}
The sensor index chosen in the $k$th step of the greedy method can be written as follows: 
\begin{align}
    i_{k}=&\argmax_{i_{k}}\,\,f_{\mathrm{EG}} , \nonumber 
\end{align}where
\begin{align}
    &f_{\mathrm{EG}}\,=\,\left\{\begin{array}{cc}
        \lambda_{\mathrm{min}}\,\left(\bm{C}_k\bm{C}_k^{\top}\right), & p\le r, \\
        \lambda_{\mathrm{min}}\,\left(\bm{C}_k^{\top}\bm{C}_k\right), & p>r.
    \end{array}\right. \nonumber \\
    &\qquad=\,\left\{\begin{array}{cc}
        \lambda_{\mathrm{min}}\,\left(
            \left[\begin{array}{cc} \bm{C}_{k-1} \\ \bm{u}_{i_k} \end{array}
            \right] \left[
            \begin{array}{cc} \bm{C}_{k-1}^{\top} &  \bm{u}_{i_k}^{\top} \end{array}
            \right]
        \right), & p\le r, \\
        \lambda_{\mathrm{min}}\,\left(
            \bm{C}_{k-1}^{\top}\bm{C}_{k-1}+\bm{u}_{i}^{\top}\bm{u}_{i}
        \right), & p>r.
    \end{array}\right.\label{eq:obj_eig_greedy}
\end{align}

\section{Submodularity and approximation rate}
\label{sec:submodularity}
In what follows, the objective functions are mathematically redefined as set functions and their structural properties are explored. Submodularity in the set functions plays an important role in combinatorial optimization and provides a lower bound of the greedy method.

A modular function has the property that each element of a subset provides an independent contribution to the function value. If the objective function is modular, it is straightforward to solve the optimization problem by the greedy method by evaluating the objective function in a step-by-step manner. On the other hand, also for monotone increasing submodular functions, which are NP-hard, the greedy method can be utilized to obtain a solution that is likely to be close to the optimal solution. 

\begin{definition}[Submodularity] 
The function $f : 2^{\{ 1, 2, \dots, n \}} \to \mathbb{R}$ is called submodular if for any  $S, T \subset \{ 1, 2, \dots, n \} $ with $S \subset T$ and $i \in \{ 1, 2, \dots, n \} \setminus T$, the function $f$ satisfies
\begin{align}
     f\bigl( S \cup \{ i \} \bigr) - f(S)
 \ge f\bigl( T \cup \{ i \} \bigr) - f(T).
 \label{eq:def_submodularity}
\end{align}
\end{definition}

\begin{definition}[Monotonicity] 
The function $f : 2^{\{ 1, 2, \dots, n \}} \to \mathbb{R}$ is called monotone increasing if for any $S, T \subset \{ 1, 2, \dots, n \} $ with $S \subset T$ and $i \in \{ 1, 2, \dots, n \} \setminus T$, the function $f$ satisfies
\begin{align}
    S \subset T \Rightarrow f(S) \le f(T).
    \label{eq:def_monotonicity}
\end{align}
\end{definition}

The performance of the greedy method is guaranteed by a well-known lower bound when the objective function is monotone and submodular. Nemhauser et al. have proved that the following inequality holds \cite{nemhauser1978analysis}:
\begin{align}
  f(S_{\mathrm{greedy}})
& \ge \left(1 - \left(1 - \frac{1}{k}\right) \right)^{k}
        f(S_{\mathrm{opt}}) \notag \\
& \ge \left(1 - \frac{1}{e} \right)
        f(S_{\mathrm{opt}}) \notag \\
& \ge 0.63 f(S_{\mathrm{opt}}),
\end{align}
where $k$ denotes the number of sensors, $S_{\mathrm{opt}}$ is an optimal solution, and $S_{\mathrm{greedy}}$ is the solution obtained from applying the greedy method.

We now evaluate the submodularity and monotonicity of objective functions introduced in the previous section.

\subsection{Objective function based on D-optimality}
\label{sec:submodularityD}
Saito et al. derived a unified expression of the objective function in both cases where $p \le r$ and $p>r$ for D-optimality. The objective function in (\ref{eq:obj_det}) is redefined to be ${\mathrm{det}}\left(\bm{C}^{\top}\bm{C}+\epsilon \bm{I}\right)$, where $\epsilon$ is a sufficiently small number. It has been proved to be the monotone submodular function\cite{saito2019determinant}.

\subsection{Objective function based on A-optimality}
\label{sec:submodularityA}
We now consider the trace of the error covariance matrix in A-optimal design. 
Similar to the previous study\cite{saito2019determinant}, we first introduce the unified formulation for both cases in which the number of sensors is less than or equal to that of the modes and the number of sensors is greater than that of the modes.
\begin{align}
    &\argmin \mathrm{tr} \left[
        \left(\bm{C}^{\top}\bm{C}+\epsilon \bm{I}\right)^{-1}
        \right] \notag \\
    =&\argmin \mathrm{tr} \left[
        -\frac{1}{\epsilon}\bm{C}\bm{C}^{\top}\left(\epsilon\bm{I}+\bm{C}\bm{C}^{\top}\right)^{-1}
        \right] + \frac{r}{\epsilon} \notag \\
    =&\argmin \mathrm{tr} \left[
        \begin{array}{l}
            -\frac{1}{\epsilon}\left(\epsilon \bm{I}+\bm{C}\bm{C}^{\top}\right)\left(\epsilon\bm{I}+\bm{C}\bm{C}^{\top}\right)^{-1}\\ \qquad\qquad\qquad + \left(\epsilon\bm{I}+\bm{C}\bm{C}^{\top}\right)^{-1}
        \end{array}\right] + \frac{r}{\epsilon} \notag \\
    =&\argmin \mathrm{tr} \left[
        -\frac{1}{\epsilon}\bm{I} +\left(\epsilon\bm{I}+\bm{C}\bm{C}^{\top}\right)^{-1}
        \right] + \frac{r}{\epsilon} \notag \\
    =&\argmin \mathrm{tr} \left[
        \left(\epsilon\bm{I}+\bm{C}\bm{C}^{\top}\right)^{-1}
        \right] + \frac{r-p}{\epsilon} \notag \\
    \approx&\argmin \mathrm{tr} \left[\left(\bm{C}\bm{C}^{\top}\right)^{-1}\right]
\end{align}
Therefore, the present objective function for the case where $p<r$ can be considered to be an asymptotic formulation of the unified objective function with a sufficiently small regularization term. 
For the proof of submodularity and monotonicity in this section, the function with the regularization term is employed as the objective function based on A-optimality hereafter.

Define a function $f_{\mathrm{A}}: 2^{ \{ 1, 2, \dots, n \} } \to \mathbb{R} $ by
\begin{align}
    f_{\mathrm{A}}(S) 
    &= - \mathrm{tr} \left[ \left( 
            \bm{C}_{S}^{\top} \bm{C}_{S}^{} + \epsilon \bm{I} 
            \right)^{-1} \right]
       + \frac{r}{\epsilon }.
    \label{eq:obj_tr_epsilon}
\end{align}
for each $S \subset \{ 1, 2, \dots, n \}$.
An offset term $r / \epsilon$ is added so that the value of $f$ for the empty set could be regarded as $f_{\mathrm{A}}(\emptyset) = 0$.
\begin{proposition} 
    $f_{\mathrm{A}}$ defined by (\ref{eq:obj_tr_epsilon}) is submodular.
    \label{thm:submodularity}
\end{proposition}
\begin{proof}
Take arbitrary $S, T \subset \{ 1, 2, \dots, n \}$ such that $S \subset T$.
For simplicity of notation, let us set 
$\bm{A} := \bm{C}_S^{\top} \bm{C}_S + \varepsilon \bm{I}$ and 
$\bm{B} := \bm{C}_S^{\top} \bm{C}_S + \bm{u}_i^{\top} \bm{u}_i + \varepsilon \bm{I}$.
It follows from the definition of $f_{\mathrm{A}}$ that, for any 
$i \in \{ 1, 2, \dots, n\} \setminus T,$
\begin{align}
    &f_{\mathrm{A}} (S \cup \{i\}) - f(S)_{\mathrm{A}} \notag \\
    &= -\mathrm{tr} \left[ \left(
                \bm{C}_S^{\top}\bm{C}_S + \bm{u}_i^{\top} \bm{u}_i
                + \epsilon \bm{I}
              \right)^{-1} \right]
      + \mathrm{tr} \left[ \left( 
      \bm{C}_S^{\top} \bm{C}_S + \epsilon \bm{I} \right)^{-1} \right] \notag \\
    &= \mathrm{tr} \left( \bm{A}^{-1} - \bm{B}^{-1} \right). \label{eq:submodular_Si-S}
\end{align}
Due to the positive definiteness of $\bm{A}$, the value of (\ref{eq:submodular_Si-S}) is positive for any $\bm{u}_{i} \neq 0$. 
We next evaluate $f(T \cup \{ i \}) - f(T)$. Since $S \subset T,$ there exits a permutation matrix $\bm{P}$ such that
\begin{equation}
   \bm{P} \bm{C}_{T} 
   = \left[  
        \begin{array}{@{\,}c@{\,}}
	  \bm{C}_{S} \\ \bm{C}_{T \setminus S} 
	\end{array}
    \right].
\end{equation}
Hence, we have
\begin{equation}
    \bm{C}_{T}^{\top} \bm{C}_{T}^{}
    = \bm{C}_{T}^{\top} \bm{P}^{\top} \bm{P} \bm{C}_{T}^{}
    = \bm{C}_{S}            ^{\top} \bm{C}_{S}^{} 
    + \bm{C}_{T \setminus S}^{\top} \bm{C}_{T \setminus S}^{},
    \label{eq:CT^TxCT2}
\end{equation}
where the fact $\bm{P}^{\top} \bm{P} = \bm{I}$ has been used.
Direct computation together with (\ref{eq:CT^TxCT2}) gives
\begin{align}
    &  f (T \cup \{i\}) - f(T) \notag \\
    =& -\mathrm{tr} \left[\left(
                \bm{C}_T^{\top}\bm{C}_T + \bm{u}_i^{\top} \bm{u}_i
                + \epsilon \bm{I} 
                \right)^{-1} \right]
       +\mathrm{tr} \left[ \left( 
                \bm{C}_T^{\top} \bm{C}_T + \epsilon \bm{I} 
                \right)^{-1} \right] \notag \\
    =& \mathrm{tr} \left[ \left(
                \bm{A} + \bm{C}^{\top}_{T\setminus S}\bm{C}_{T\setminus S} 
                \right)^{-1} \right]
      -\mathrm{tr} \left[ \left(
                \bm{B} + \bm{C}^{\top}_{T\setminus S}\bm{C}_{T\setminus S} 
                \right)^{-1} \right] \notag \\
    =& \mathrm{tr} \left[ 
                \bm{A}^{-1} - \bm{A}^{-1} \bm{C}^{\top}_{T\setminus S}
                \left(
                    \bm{I} + \bm{C}_{T\setminus S} \bm{A}^{-1}\bm{C}^{\top}_{T\setminus S}
                \right)^{-1}
                \bm{C}_{T\setminus S} \bm{A}^{-1} \right] \notag \\
     & -\mathrm{tr} \left[ 
                \bm{B}^{-1} - \bm{B}^{-1} \bm{C}^{\top}_{T\setminus S}
                \left(
                    \bm{I} + \bm{C}_{T\setminus S} \bm{B}^{-1}\bm{C}^{\top}_{T\setminus S}
                \right)^{-1}
                \bm{C}_{T\setminus S} \bm{B}^{-1} \right].
    \label{eq:submodular_Ti-T}
\end{align}
Therefore,
\begin{align}
    &f (S \cup \{i\}) - f(S)- f (T \cup \{i\}) + f(T) \notag \\
    =& \mathrm{tr} \left[
                \bm{A}^{-2} \bm{C}^{\top}_{T\setminus S}
                \left(
                    \bm{I} + \bm{C}_{T\setminus S} \bm{A}^{-1}\bm{C}^{\top}_{T\setminus S}
                \right)^{-1}
                \bm{C}_{T\setminus S} \right] \notag \\
     & -\mathrm{tr} \left[ 
                \bm{B}^{-2} \bm{C}^{\top}_{T\setminus S}
                \left(
                    \bm{I} + \bm{C}_{T\setminus S} \bm{B}^{-1}\bm{C}^{\top}_{T\setminus S}
                \right)^{-1}
                \bm{C}_{T\setminus S} \right] \notag \\
    \ge& \mathrm{tr} \left[
                \bm{A}^{-2} \bm{C}^{\top}_{T\setminus S}
                \left(
                    \bm{I} + \bm{C}_{T\setminus S} \bm{A}^{-1}\bm{C}^{\top}_{T\setminus S}
                \right)^{-1}
                \bm{C}_{T\setminus S} \right] \notag \\
     & -\mathrm{tr} \left[ 
                \bm{B}^{-2} \bm{C}^{\top}_{T\setminus S}
                \left(
                    \bm{I} + \bm{C}_{T\setminus S} \bm{A}^{-1}\bm{C}^{\top}_{T\setminus S}
                \right)^{-1}
                \bm{C}_{T\setminus S} \right] \notag \\
    =& \mathrm{tr} \left[ 
                \left(\bm{A}^{-2} - \bm{B}^{-2} \right) 
                \bm{C}^{\top}_{T\setminus S} 
                \left(
                    \bm{I} + \bm{C}_{T\setminus S} \bm{A}^{-1}\bm{C}^{\top}_{T\setminus S}
                \right)^{-1}
                \bm{C}_{T\setminus S} \right],
    \label{eq:submodular_Si-S-Ti+T}
\end{align}
where the fact that the trace of a positive semidefinite matrix is nonnegative and the following semidefiniteness are adopted for the derivation of the inequality of the second to third equation. 
\begin{align}
    &\bm{C}^{\top}_{T\setminus S}
                \left(
                    \bm{I} + \bm{C}_{T\setminus S} \bm{A}^{-1}\bm{C}^{\top}_{T\setminus S}
                \right)^{-1}
                \bm{C}_{T\setminus S} \notag \\
    & - \bm{C}^{\top}_{T\setminus S}
                \left(
                    \bm{I} + \bm{C}_{T\setminus S} \bm{B}^{-1}\bm{C}^{\top}_{T\setminus S}
                \right)^{-1}
                \bm{C}_{T\setminus S} \succeq 0.
\end{align}
Since
\begin{align}
    \bm{A}^{-2} - \bm{B}^{-2} \succeq 0
\end{align}
and 
\begin{align}
    \bm{C}^{\top}_{T\setminus S} 
    \left(
        \bm{I} + \bm{C}_{T\setminus S} \bm{A}^{-1}\bm{C}^{\top}_{T\setminus S}
    \right)^{-1}
    \bm{C}_{T\setminus S} \succeq 0,
\end{align}
the trace of the product of a positive and nonnegative semidefinite matrix is nonnegative.
This completes the proof.
\end{proof}

\begin{proposition} 
    $f_{\mathrm{A}}$ defined by (\ref{eq:obj_tr_epsilon}) is monotone increasing.
    \label{thm:monotonicity}
\end{proposition}

\begin{proof}
For given $S, T \subset \{ 1, 2, \dots, n \}$ with $S \subset T$, it is clear that (\ref{eq:submodular_Si-S}) holds also for any $i \in T \setminus S $.
Hence, $f(S \cup \{ i \}) - f(S) \ge 0$.
Similarly, we can show that 
$f(S \cup \{ i \} \cup \{ j \}) - f(S) \ge 0$
for $j \in T \setminus  (S \cup \{ i\}) $.
Repeated application of this argument yields
$f(S) \le f(T)$, which is the desired conclusion.
This completes the proof.
\end{proof}

\subsection{Objective function based on E-optimality}
\label{sec:submodularityE}
We demonstrate the property of the objective function based on E-optimality, given by (\ref{eq:obj_eig}).
The objective function returns the minimum eigenvalue of the Fisher information matrix.

We first show by counterexample that the objective function fails to be submodular. 
If it is submodular, adding a sensor to a subset where a smaller number of sensors is selected gives a higher increment in the minimum eigenvalue than adding one to a subset where a larger number of sensors is selected.
Let us set $r$ = 3 and $n$ = 10, and consider a sensor candidate matrix $\bm{U}$ defined by 
\begin{align}
    \bm{U} = \left[
        \begin{array}{c}
        \bm{u}_1\\
        \bm{u}_2\\
        \bm{u}_3\\
        \bm{u}_4\\
        \bm{u}_5\\
        \bm{u}_6\\
        \end{array}
    \right]
    =
     \left[
        \begin{array}{rrr}
        0.2 & -0.1 & -0.2 \\
        -0.5 & -0.1 & 0.2 \\
        -0.2 & 0.3 & 0.2 \\
        -0.5 & 0.3 & -0.3 \\
        -0.4 & -0.3 & -0.4 \\
        0.3 & 0 & 0 \\
        \end{array}
    \right]. \notag
\end{align}
Suppose the measurement matrices where 3, 4, and 5 sensors are selected given by
\begin{align}
    &\bm{C}_{3} 
    = \left[ 
      \begin{array}{ccc}
        \bm{u}_1^{\top} &
        \bm{u}_2^{\top} &
        \bm{u}_3^{\top}
      \end{array}
   \right]^{\top}, \notag \\
    &\bm{C}_{4} 
    = \left[ 
      \begin{array}{cccc}
        \bm{u}_1^{\top} &
        \bm{u}_2^{\top} &
        \bm{u}_3^{\top} &
        \bm{u}_4^{\top}
      \end{array}
   \right]^{\top}, \notag \\
   &\bm{C}_{5} 
    = \left[ 
      \begin{array}{ccccc}
        \bm{u}_1^{\top} &
        \bm{u}_2^{\top} &
        \bm{u}_3^{\top} &
        \bm{u}_4^{\top} &
        \bm{u}_5^{\top}
      \end{array}
   \right]^{\top}. \notag
\end{align}
Then, comparing the increment in the minimum eigenvalue by adding a sensor, we obtain
\begin{align}
    &\lambda_{\mathrm{min}} \left( \bm{C}_{3p} \right)
    - \lambda_{\mathrm{min}} \left( \bm{C}_3 \right)
    >
    \lambda_{\mathrm{min}} \left( \bm{C}_{4p} \right)
    - \lambda_{\mathrm{min}} \left( \bm{C}_4 \right), 
    \label{eq:submodular_add5th}
\end{align}
where 
\begin{align}
    &\bm{C}_{3p} 
    = \left[ 
      \begin{array}{c}
        \bm{C}_3 \\
        \bm{u}_5
      \end{array}
    \right], 
    \bm{C}_{4p} 
    = \left[ 
      \begin{array}{c}
        \bm{C}_4 \\
        \bm{u}_5
      \end{array}
    \right], \notag
\end{align}
and
\begin{align}
    &\lambda_{\mathrm{min}} \left( \bm{C}_{4p'} \right)
    - \lambda_{\mathrm{min}} \left( \bm{C}_4 \right)
    <
    \lambda_{\mathrm{min}} \left( \bm{C}_{5p'} \right)
    - \lambda_{\mathrm{min}} \left( \bm{C}_6 \right),
    \label{eq:submodular_add6th}
\end{align}
where
\begin{align}
    &\bm{C}_{4p'} 
    = \left[ 
      \begin{array}{c}
        \bm{C}_4 \\
        \bm{u}_6
      \end{array}
    \right], 
    \bm{C}_{5p'} 
    = \left[ 
      \begin{array}{c}
        \bm{C}_5 \\
        \bm{u}_6
      \end{array}
    \right]. \notag
\end{align}
Since (\ref{eq:submodular_add5th}) and (\ref{eq:submodular_add6th}) correspond to submodularity and supermodularity, respectively, the objective function based on E-optimality is neither submodular, supermodular, nor modular. 

We now turn to the monotonicity of (\ref{eq:obj_eig}).
There is a well-known theorem which shows the lower bound in the minimum eigenvalue by rank-one modification of the Hermitian matrix\cite{cheng2012bounds}.  
According to the theorem,  the objective function based on E-optimality is monotone increasing in the case of $p>r$.
Note that the minimum value of the eigenvalues of $\bm{C}\bm{C}^{\top}$ is not monotone in the case of $p \le r$ because if an eigenvalue is newly added into the system it can be lower or higher than the minimum eigenvalue of the previous subset.

\section{Results and Discussion}
\label{sec:results}
In this section, the performance of sparse sensor selection methods based on different optimality criteria and schemes is evaluated. The results of numerical experiments on randomly generated systems and on a practical dataset of global ocean surface temperature are illustrated.

The performance of each method depends on the characteristics of the dataset and situation. Therefore, numerical experiments under two different conditions are conducted and guidance on selecting a proper objective function is put forward for these differing contexts. These problems are the minimum set, but they clearly show the characteristics of the greedy methods with three sensor selection criteria adopted in the present study, i.e., D-optimality, A-optimality, and E-optimality. The first problem is sensor selection for the estimation of the latent state variables given by the Gaussian distribution on randomly generated systems, which is the ideal case where the statistical assumption for the system is valid. The second problem is conducted using a practical dataset concerning global ocean surface temperature, and it focuses on sensor selection for the estimation of mode strengths as latent state variables under the situation where the test data are different from the training data. It corresponds to the nonideal case where the statistical assumption for the system is not valid as in the cross-validation study. 

Four methods are compared: the D-optimality-based greedy (DG) method, the A-optimality-based greedy (AG) method, the E-optimality-based greedy (EG) method, and the D-optimality-based convex relaxation (DC) method listed in Table \ref{table:Calc_method}. The comparison between the DG, AG, and EG methods provides insights into the characteristics of three sensor selection criteria when they are used in conjunction with the greedy method. 
The DG method which is efficient in terms of computational time was proposed by Saito et al.\cite{saito2019determinant}. In the DG method, the sensors are determined by maximizing the determinant of (\ref{eq:obj_det}) for each step in the case where $p>r$, and on the other hand, QR pivoting\cite{manohar2018data} is employed in the case where $p \le r$ since QR implementation for the greedy method is much faster than the straightforward implementation of the greedy method maximizing (\ref{eq:obj_det}). 
Furthermore, ``eig'' and ``min'' functions in MATLAB R2020a are employed and the minimum eigenvalue is obtained in the EG method.
The DC method corresponds to the D-optimality-based convex relaxation method proposed by Joshi and Boyd\cite{joshi2009sensor}. The MATLAB code is available on Github\cite{nakai2020github}. 
The numerical experiments are conducted under the computational environment listed in Table \ref{table:Comp_env}.

\begin{center}
\begin{threeparttable}[ht]
    \small
    \centering
    \caption{Sensor selection methods investigated in this study} 
    \label{table:Calc_method}
    \begin{tabular}[t]{lll}
        \hline\hline
        Name & Optimality & Scheme\\
        \hline\hline
        DG\cite{saito2019determinant} & D & ${p \le  r}$ : Greedy based on QR \\
        & & ${p>r}$ : Greedy \\
        \hline
        AG & A & Greedy (Eq. (\ref{eq:obj_tr_greedy}))\\
        \hline
        EG & E & Greedy (Eq. (\ref{eq:obj_eig_greedy}))\\
        \hline
        DC\cite{joshi2009sensor} & D & Convex relaxation\\
        \hline\hline
    \end{tabular}
\end{threeparttable}
\end{center}

\begin{center}
\begin{threeparttable}[ht]
    \small
    \centering
    \caption{Computational environment} 
    \label{table:Comp_env}
    \begin{tabular}[t]{ll}
        \hline\hline
        Processor information & Intel(R) Core(TM)\\
        & i7-10510U@ 1.80 GHz \\
        \hline
        Random access memory & 16 GB\\
        \hline
        System type & 64 bit operating system\\
        & x64 base processor\\
        \hline
        Operating system & Windows 10 Home\\
        & Version:2001\\
        \hline
        Source code & MATLAB R2020a\\
        \hline\hline
    \end{tabular}
\end{threeparttable}
\end{center}

\subsection{Performance on random systems}
\label{sec:results-random}
In this subsection, randomly generated data are considered and the performance of the sparse sensor selection methods listed in Table \ref{table:Calc_method} is evaluated. 
The random sensor-candidate matrices, $\bm{U}\in \mathbb{R}^{n{\times}r}$, are set where the component of the matrices is given by the Gaussian distribution of $\mathcal{N}(0,1)$ with $n$ = 2000 and $r$ = 10, and the $p$ sparse sensors are selected by each method. Then, the values of indices utilized for optimal criteria, i.e. the determinant, the trace of the inverse, and the minimum eigenvalue of the Fisher information matrix are evaluated. In addition, the latent state variables are reconstructed using the sparse sensors selected by each method in order to compare the estimation accuracy. Each component of the latent state variables vector, $\bm{z}\in \mathbb{R}^{r{\times}m}$, is given by the Gaussian distribution of $\mathcal{N}(0,1)$ with $r$ = 10 and $m$ = 1 in this error evaluation test. The computational time required to obtain the sparse sensors by each method is also compared. 

First, the optimal indices for sensors selected by each method are demonstrated. Figs.~\ref{fig:random_determinant}, \ref{fig:random_trace}, and \ref{fig:random_eigenvalue} show the relationship between the number of sensors and three indices: the determinant, the trace of the inverse, and the minimum eigenvalue of the Fisher information matrix, that is, $\bm{CC}^{\top}=\bm{HUU}^{\top}\bm{H}^{\top}$ ($p \le r$) and $\bm{C}^{\top}\bm{C}=\bm{U}^{\top}\bm{H}^{\top}\bm{HU}$ ($p > r$), respectively. The results of the DG, AG, EG, DC, and random selection methods are plotted together for comparison. Note that the results obtained from the random selection method are not plotted in Figs. 3 and 4 because the values of the trace of the inverse and minimum eigenvalue are larger and smaller than the range of the figures, respectively. The values plotted in Figs.~\ref{fig:random_determinant}, \ref{fig:random_trace}, and \ref{fig:random_eigenvalue} are averages over 2000 random samples and normalized by those of the DG method since this study focuses particulary on evaluating methods based on A- and E-optimality in comparison with that based on D-optimality.

Fig.~\ref{fig:random_determinant} reveals that the determinant obtained by the DG method is higher than that obtained by the AG and EG methods for any number of sensors. This indicates that the greedy method utilizing the objective function based on D-optimality is the most suitable for the maximization of the determinants, as expected. Fig.~\ref{fig:random_determinant} also demonstrates that the sensors selected by the AG method are superior to those selected by the EG method in terms of maximizing the determinant. 

Fig.~\ref{fig:random_trace} shows that the sensors selected by the AG method work beetter for minimizing the trace of the inverse in the case where $p \le$ 12 compared to the DG and EG methods. The AG method works as well as the DG method in the case where $p >$ 12. Note that the trace of the inverse obtained by the AG method, which is dedicated to minimizing the trace of the inverse, is slightly higher than that obtained by the DG method when $p$ exceeds 12. This is mainly because the greedy method does not obtain a combinatorial optimized solution but a suboptimal one in a step-by-step manner. On the other hand, the sensors selected by the EG method have significantly higher trace of the inverse in the case where $p >$ 10 compared to those according to the AG and DG methods. Therefore, the greedy methods based on A- and D-optimality are more suitable for minimizing the trace of the inverse. 

Fig.~\ref{fig:random_eigenvalue} shows that the sensors selected by the EG method, which seeks to maximize the minimum eigenvalue, work best for maximizing the minimum eigenvalue in the case where $p \le$ 11. However, the minimum eigenvalue obtained by the EG method is lower than that obtained by the DG and AG methods in the case where $p >$ 12. This might be due to the fact that the objective function based on E-optimality is not submodular; thus, its combination with the greedy method does not lead to good results, as stated in Section \ref{sec:submodularity}. 
On the other hand, the AG method works as well as the EG method in the case where $p \le$ 11, and the sensors selected by the AG method have the highest minimum eigenvalue in the case where $p >$ 11. This might be because the original objective function of the AG method is more related with maximization of the minimum eigenvalue of the Fisher information matrix than that of the DG method, as discussed in the next paragraph. Therefore, in contrast to the EG method, the AG method utilizing the submodular objective function performs well in terms of maximizing the minimum eigenvalue of the Fisher information matrix. Accordingly, the greedy method based on A-optimality and, depending on the conditions, the greedy method based on E-optimality are suitable for the selection of sensors which have higher minimum eigenvalues. Further, compared to the DG, AG, and EG methods, the DC method works well on the determinant, but not on the trace of the inverse or the minimum eigenvalue. Although, in broad terms, this can be attributed to the convex relaxation formulation, the detailed reason requires further investigation.

Here, the reason why the objective function of the AG method has stronger relationship with that of the EG method than that of the DG method is qualitatively explained. The objective function of the DG, AG and EG methods can be written for the oversampling situation as follows:
\begin{align}
f_{\textrm{D}}&=\lambda_1 \times \lambda_2 \times \cdots \times \lambda_r, \\
f_{\textrm{A}}&=\frac{1}{\lambda_1} + \frac{1}{\lambda_2} + \cdots + \frac{1}{\lambda_r}, \\
f_{\textrm{E}}&=\lambda_r,
\end{align}
where $\lambda_i$ is the $i$th largest eigenvalue of $\bm{C}^{\top}\bm{C}$ and $\lambda_r=\lambda_{\text{min}}$. A rational increase in any eigenvalue works evenly for the improvement in the objective function of the DG method, and therefore, the increase in minimum eigenvalue $\lambda_r$ is not necessarily required for the improvement if there is a sensor selection which increases other eigenvalues more in ratio. On the other hand, the increase in $\lambda_r$ is more important for the objective function of the AG method because 
$1 / \lambda_r = 1/ f_{\textrm{E}}$ is the largest term in $f_{\textrm{A}}$. Therefore, the increase in $\lambda_r$ is strongly demanded in the minimization of $f_{\textrm{A}}$ in the AG method. 

Next, the performance of the greedy methods based on different objective functions is discussed from the perspective of the reconstruction error. Fig.~\ref{fig:random_error} shows the relationship between the number of sensors and reconstruction error, where the reconstruction error of latent state variables is defined by the ratio of the difference between reconstructed value and true value to the true value. Note that results from the random selection method are not plotted in Fig. 5 because the reconstruction error in this case is too large and exceeds the range of the figure. The reconstruction error decreases as $p$ increases for the sensors selected by all the methods. The reconstruction error using the sensors selected by the DG and AG methods is lower than that by the EG method, and that by the DG method is slightly lower than that by the AG method except for a couple of conditions. This indicates that the determinant and the trace of the inverse of the Fisher information matrix, which correspond to the determinant of the error covariance matrix and mean square error, respectively, are important indices for accurate reconstruction in the case where the statistical assumption for the system is valid. 

Fig.~\ref{fig:random_time} illustrates the computational time required for sensor selection using the DG, AG, EG, DC, and random selection methods. 
The computational time for the AG and EG methods gradually increases with increasing $p$ because the greedy method chooses sensors in a step-by-step manner. The AG and EG methods need less time than the DC method. This indicates that all the greedy methods based on D-, A-, and E-optimality are superior to the convex relaxation method in terms of computational time. On the other hand, the AG and EG methods require more time than the DG method. This is mainly because the DG method utilizes QR decomposition in the case where $p \le r$. Although the AG method seems to be comparable to the DG method in terms of the indices of selected sensor quality, considering the computational cost, the latter is more effective for sensor selection in the ideal case where the statistical assumption for the system is valid.

\begin{figure}[htbp]
    \centering
    \includegraphics[width=3in]{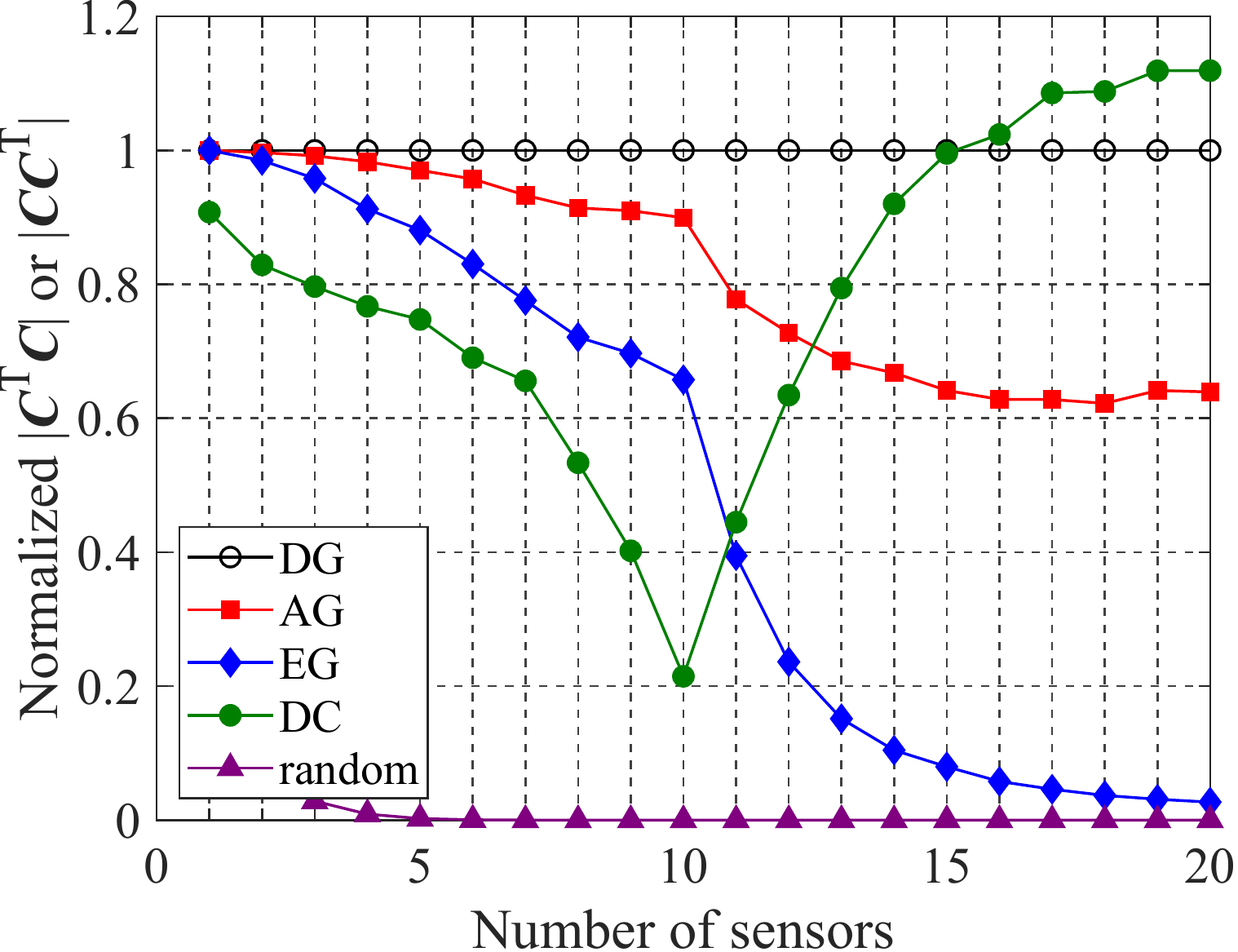}
    \caption{Normalized determinant of $\bm{CC}^{\top}$ ($p \le r$) or $\bm{C}^{\top}\bm{C}$ ($p > r$) against the number of sensors for random systems.}
    \label{fig:random_determinant}
\end{figure}

\begin{figure}[htbp]
    \centering
    \includegraphics[width=3in]{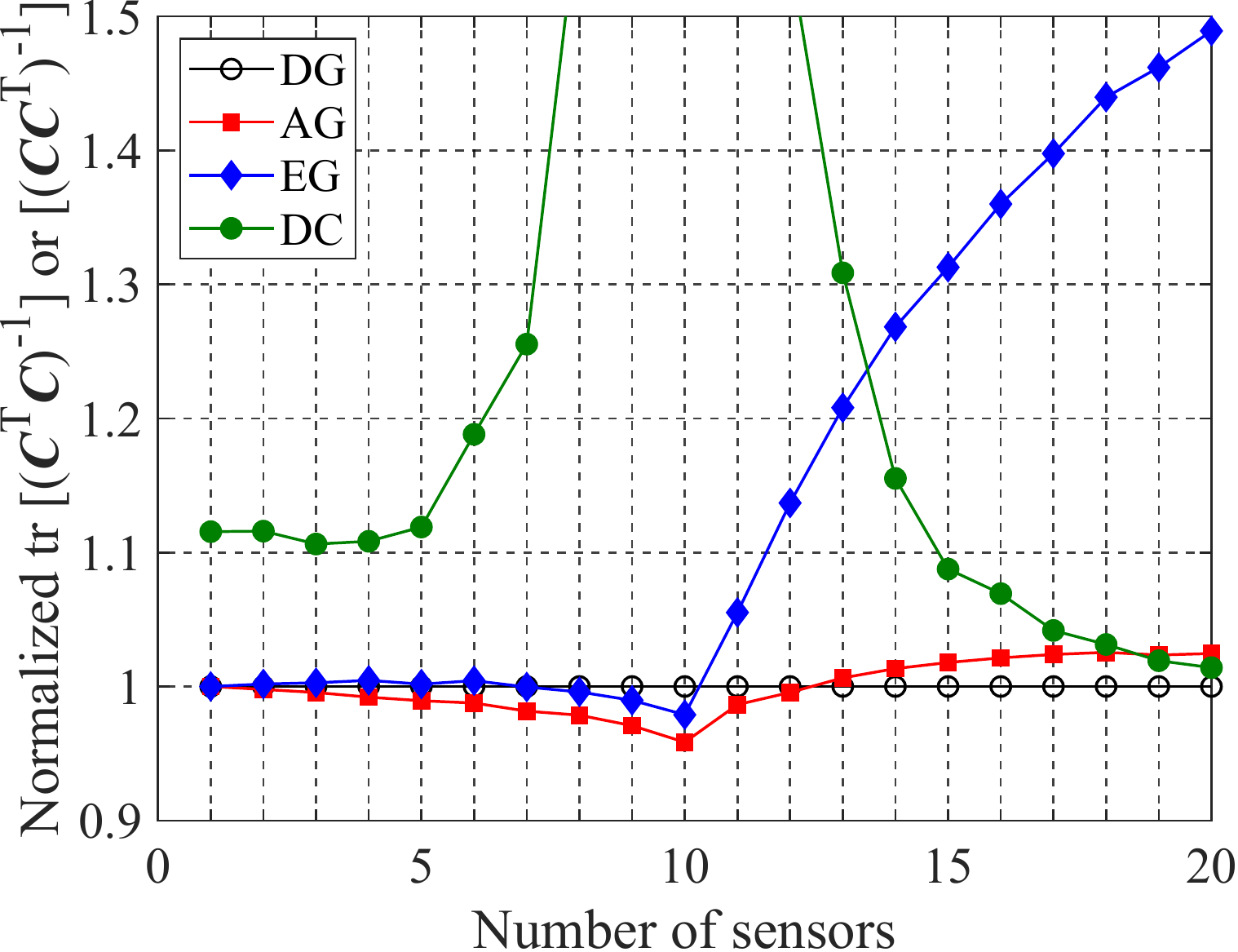}
    \caption{Normalized trace of $\left(\bm{CC}^{\top}\right)^{-1}$ ($p \le r$) or $\left(\bm{C}^{\top}\bm{C}\right)^{-1}$ ($p > r$) against the number of sensors for random systems. The results of the random selection are omitted because they are larger than the plotted range.}
    \label{fig:random_trace}
\end{figure}

\begin{figure}[htbp]
    \centering
    \includegraphics[width=3in]{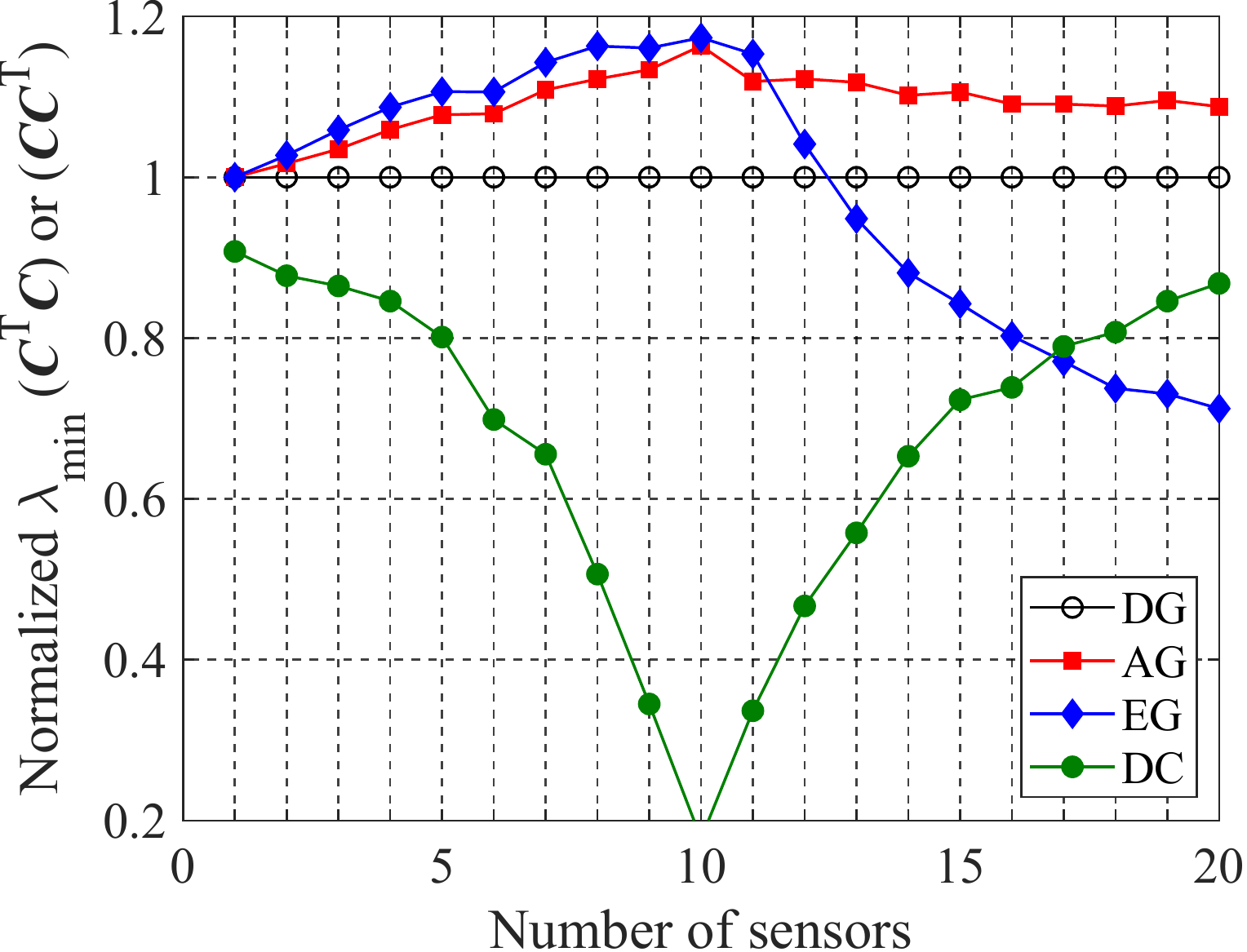}
    \caption{Normalized minimum eigenvalue of $\bm{CC}^{\top}$ ($p \le r$) or $\bm{C}^{\top}\bm{C}$ ($p > r$) against the number of sensors for random systems. 
    The results of the random selection are omitted because they are smaller than the plotted range.}
    \label{fig:random_eigenvalue}
\end{figure}

\begin{figure}[htbp]
    \centering
    \includegraphics[width=3in]{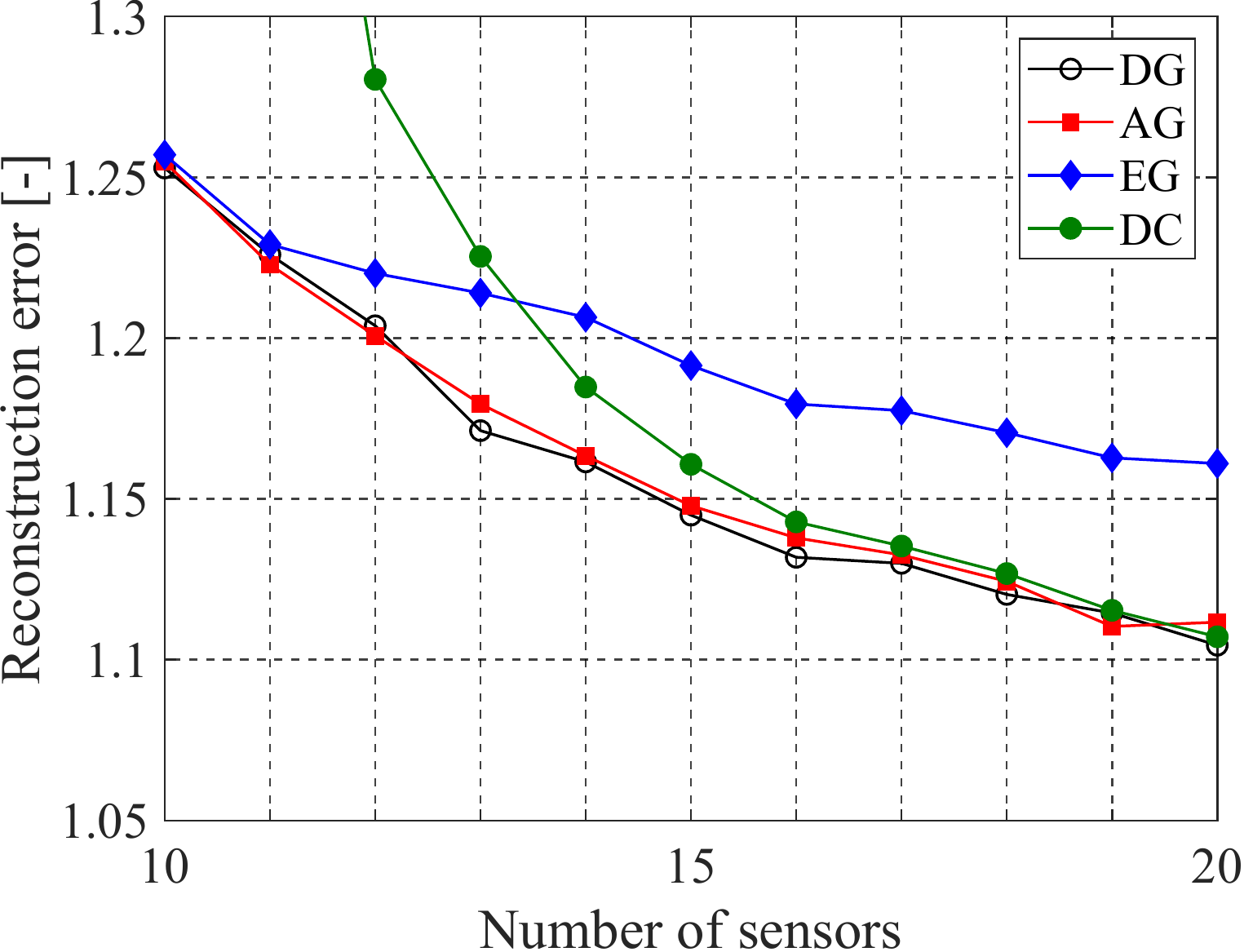}
    \caption {Reconstruction error against the number of sensors for the random system.The results of the random selection are omitted because they are larger than the plotted range.}
    \label{fig:random_error}
\end{figure}

\begin{figure}[htbp]
    \centering
    \includegraphics[width=3in]{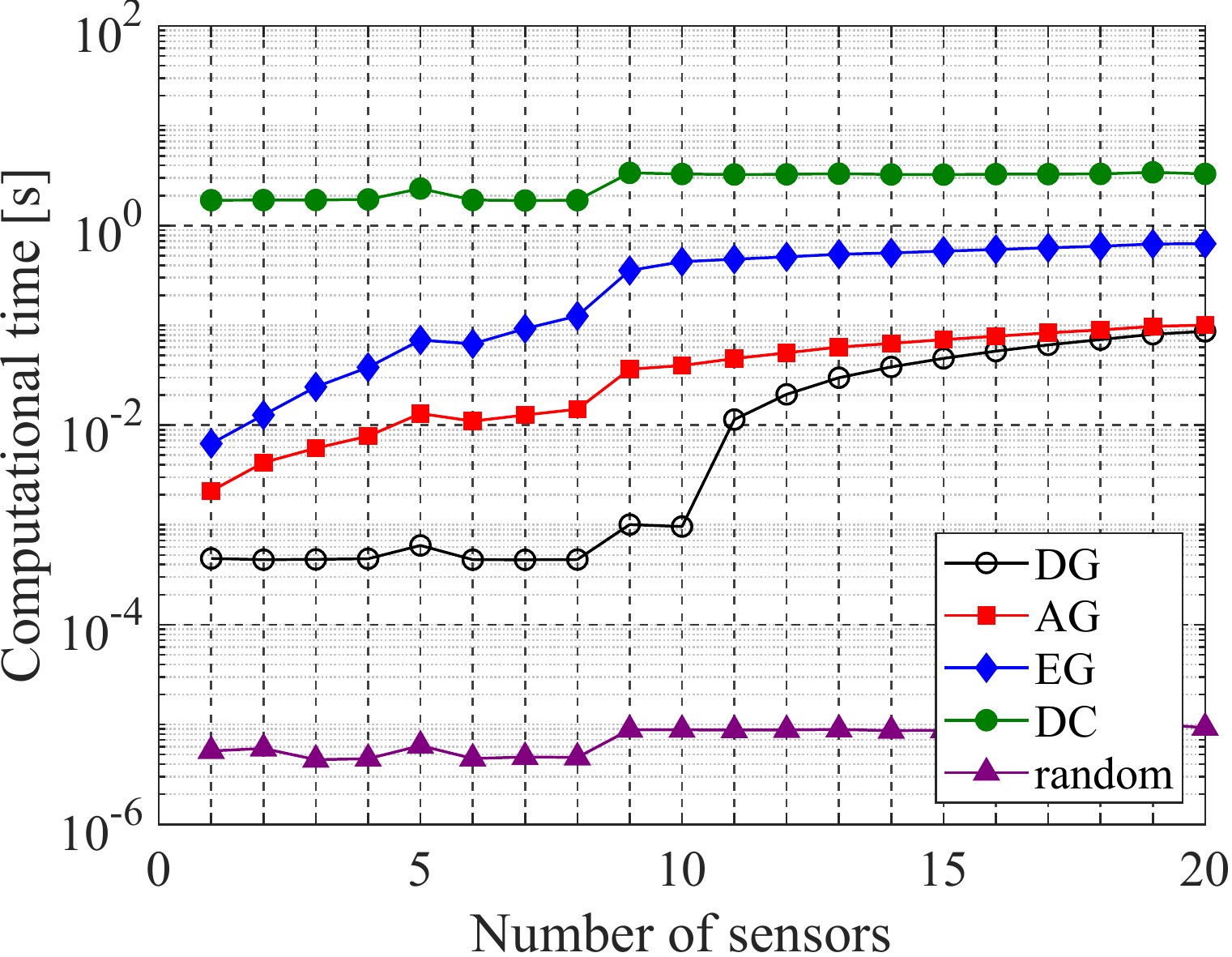}
    \caption{Computational time against the number of sensors for random systems.}
    \label{fig:random_time}
\end{figure}

\subsection{Performance on K-fold cross-validation in NOAA-SST problem}
In this subsection, $K$-fold cross-validation is conducted for sensor selection of the data-driven reconstruction problem of the NOAA OISST (NOAA-SST) V2 mean sea surface temperature dataset. Therefore, the performance of the sensors selected by the DG, AG, and EG methods is evaluated in the case where the statistical assumptions of equally distributed independent Gaussian noise is not exactly valid, which seems to be a more practical situation of the cross-validation. Note that the DC method is not applied to this problem because of its high computational cost for the cross-validation as shown in Fig.~\ref{fig:random_time}. Numerical tests using the same data as that for the training (that is not a cross-validation study) have also been conducted, and as a result, the indices and the reconstruction error of each method have intermediate characteristics between those for the random system problem discussed in Subsection \ref{sec:results-random} and those for the $K$-fold cross-validation discussed in this subsection. In addition, numerical tests of the $K$-fold cross-validation with $K$ = 3 and 7 have been conducted. Although the results change quantitatively by varying the $K$ parameter, the qualitative performance of each method does not depend on the value of this parameter. Furthermore, we have evaluated the performance on another practical dataset related to the flowfield around an airfoil, and obtained results similar to those for the cross-validation in the NOAA-SST problem. These results are omitted here for brevity.

The NOAA SST dataset consists of weekly global sea surface temperature measurements in the years between 1990 and 2000\cite{noaa}. Dimensional reduction is conducted for this problem in exactly the same way as in a previous study\cite{saito2019determinant}. The data matrix $\bm{X}\in\mathbb{R}^{n\times m}$, which consists of $m$ snapshots of the spatial temperature distribution vector of dimension $n$, is decomposed into the left singular matrix $\bm{U} \in \mathbb{R}^{n\times m}$, the diagonal matrix of singular values $\bm{S} \in \mathbb{R}^{m\times m}$, and the right singular matrix $\bm{V} \in \mathbb{R}^{m\times m}$, by economy SVD, corresponding to POD.
Here, $\bm{U}$, $\bm{S}$, and $\bm{V}$ show the spatial POD modes, the POD mode amplitudes, and the temporal POD modes, respectively. 
Hence, the reduced-order model $\bm{X}\approx\bm{U}_{1:r}\bm{S}_{1:r}\bm{V}_{1:r}^{\top}$ is given by applying the truncated SVD for a given rank $r$. 
Then, the measurement matrix is described by $\bm{C}=\bm{H}\bm{U}_{1:r}$, and the latent state variables become the POD mode amplitude: $\bm{Z}=\bm{S}_{1:r}\bm{V}_{1:r}^{\top}$.
Refer to \cite{saito2019determinant} for detailed information. 
The NOAA-SST dataset has 520 snapshots on a 360 $\times$ 180 spatial grid, and the data truncated to the $r=10$ POD modes are used in this study. The 520 snapshots are partitioned into 5 segments ($K$=5). One segment is used as the test data with the remaining segments used for the training data. 

Fig.~\ref{fig:NOAA_map} shows the sensor locations obtained by using the DG, AG, and EG methods, respectively, in the case where $p=15$ of one test. Although the locations selected are slightly different depending on the method, they are similar, especially with respect to the DG and AG methods.

The results of optimal indices are shown in Figs.~\ref{fig:NOAA_determinant}, \ref{fig:NOAA_trace}, and \ref{fig:NOAA_eigenvalue}: the determinant, the trace of the inverse, and the minimum eigenvalue of the Fisher information matrix, respectively. Figs.~\ref{fig:NOAA_determinant}, \ref{fig:NOAA_trace}, and \ref{fig:NOAA_eigenvalue} illustrate that the characteristics of indices are similar to those for the random system problem discussed in Subsection \ref{sec:results-random}, though plots are noisy because results pertain to only five tests of sensor selection. Fig.~\ref{fig:NOAA_determinant} shows that the determinants obtained by the DG method are the highest, and those by the EG method are the lowest with the AG method being intermediate for any number of sensors, as with the random system problem. The sensors selected by the DG and AG methods also work better for minimizing the trace of the inverse compared to those by the EG method in the case of large $p$, as shown in Fig.~\ref{fig:NOAA_trace}. Note that, unlike the results of the random system problem, the sensors selected by the AG method are marginally superior to those by the DG method for any number of sensors. Fig.~\ref{fig:NOAA_eigenvalue} shows that the sensors selected by the EG method work best for maximizing the minimum eigenvalue in the case where $p \le$ 12, and the sensors selected by the AG method tend to have a high minimum eigenvalue regardless of the number of sensors, similar to the random system problem. Note that the sensors selected by the EG method work better for larger $p$ than in the random system problem. 

Fig.~\ref{fig:NOAA_error_CV} shows the relationship between the reconstruction error and the number of sensors. In contrast to the random system problem shown in Subsection \ref{sec:results-random}, the error associated with the DG method is higher than that of the AG and EG methods, and the AG method has a lower error than the EG method. This might be because the worst-case variance of the estimation error governs the performance in the case where the statistical assumption for the system is not valid. This corresponds to E-optimality, which is the objective function of the EG method. Accordingly, the sensors selected by the EG and AG methods, which have a high minimum eigenvalue of the Fisher information matrix as demonstrated in Fig.~\ref{fig:NOAA_eigenvalue}, perform better in terms of the estimation error than those by the DG method. In addition, the sensors selected by the AG method, which tend to have a high minimum eigenvalue for any number of sensors as shown in Fig.~\ref{fig:NOAA_eigenvalue}, work better than the EG method. 
Therefore, the greedy method with the objective function based on A-optimality is the most suitable for reconstruction of practical problems, in which the statistical assumption for the noise involved is not exactly valid, because of the higher minimum eigenvalue of the Fisher information matrix of the sensors selected by the AG method rather than that by the EG method. 

Fig.~\ref{fig:NOAA_time} shows the computational time required to obtain the sensors. The qualitative characteristics in Fig.~\ref{fig:NOAA_time} are the same as those in Fig.~\ref{fig:random_time}. 
Therefore, considering the reconstruction error and computational cost comprehensively, the greedy method based on A-optimality is the most suitable for selecting the appropriate sensors for the reconstruction of the practical dataset.

\begin{figure}[htbp]
    \centering
    \includegraphics[width=3in]{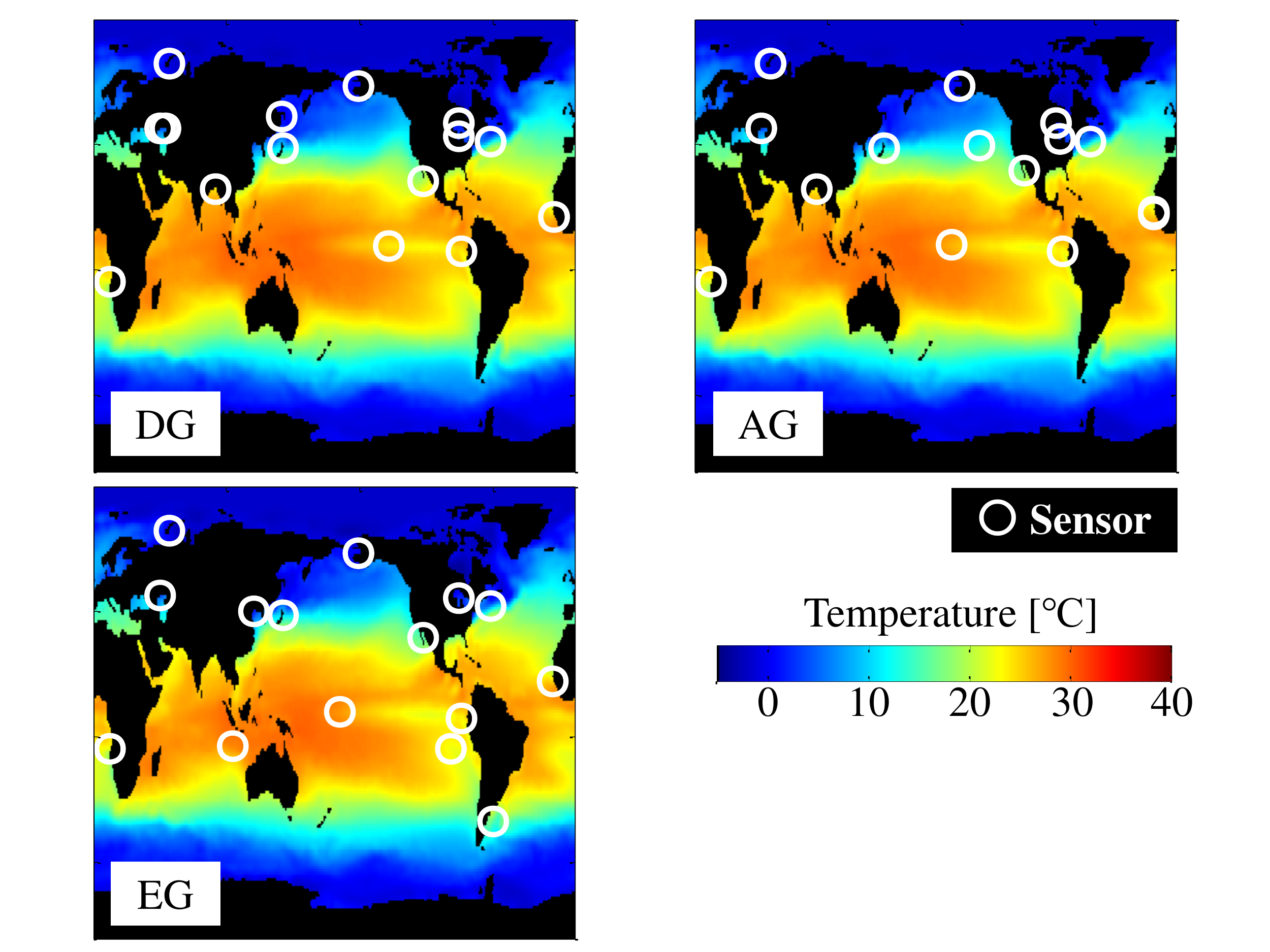}
    \caption {Sensor locations in the case where $p=15$ for the NOAA-SST problem.}
    \label{fig:NOAA_map}
\end{figure}

\begin{figure}[htbp]
    \centering
    \includegraphics[width=3in]{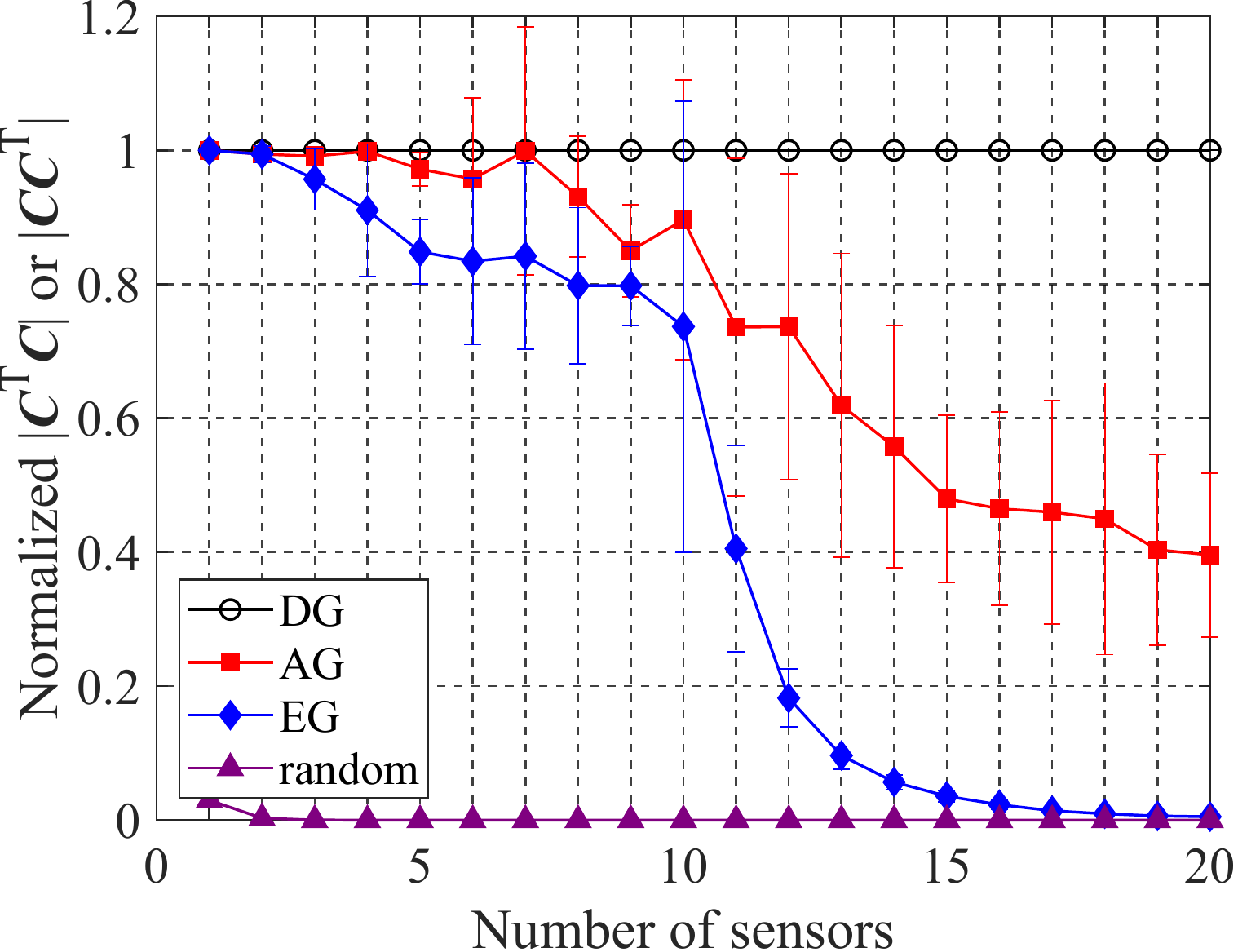}
    \caption {Normalized determinant of $\bm{CC}^{\top}$ ($p \le r$) or $\bm{C}^{\top}\bm{C}$ ($p > r$) against the number of sensors for the NOAA-SST problem. The results of the random selection are omitted because they are larger than the plotted range.}
    \label{fig:NOAA_determinant}
\end{figure}

\begin{figure}[htbp]
    \centering
    \includegraphics[width=3in]{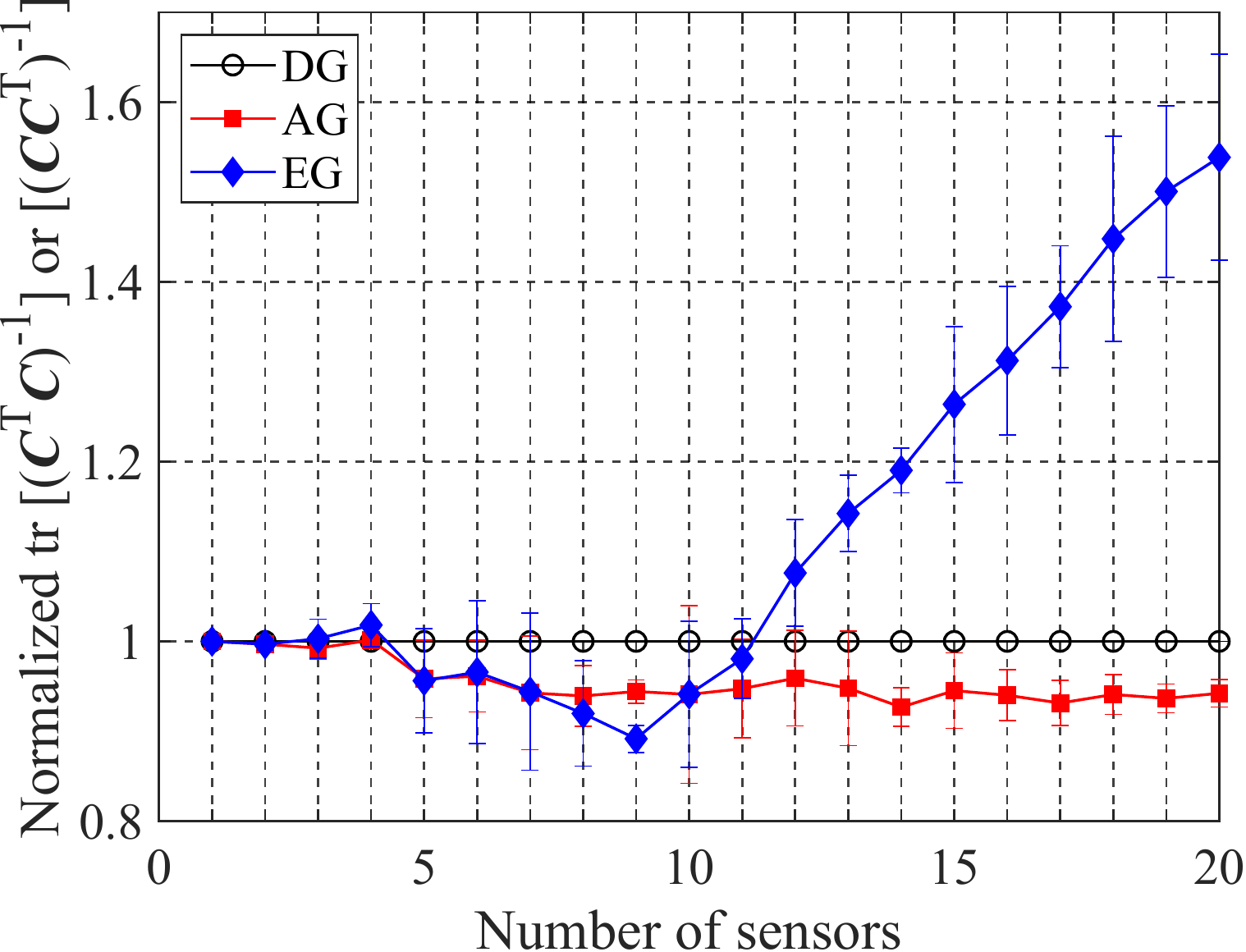}
    \caption {Normalized trace of $\left(\bm{CC}^{\top}\right)^{-1}$ ($p \le r$) or $\left(\bm{C}^{\top}\bm{C}\right)^{-1}$ ($p > r$) against the number of sensors for the NOAA-SST problem. The results of the random selection are omitted because they are smaller than the plotted range.}
    \label{fig:NOAA_trace}
\end{figure}

\begin{figure}[htbp]
    \centering
    \includegraphics[width=3in]{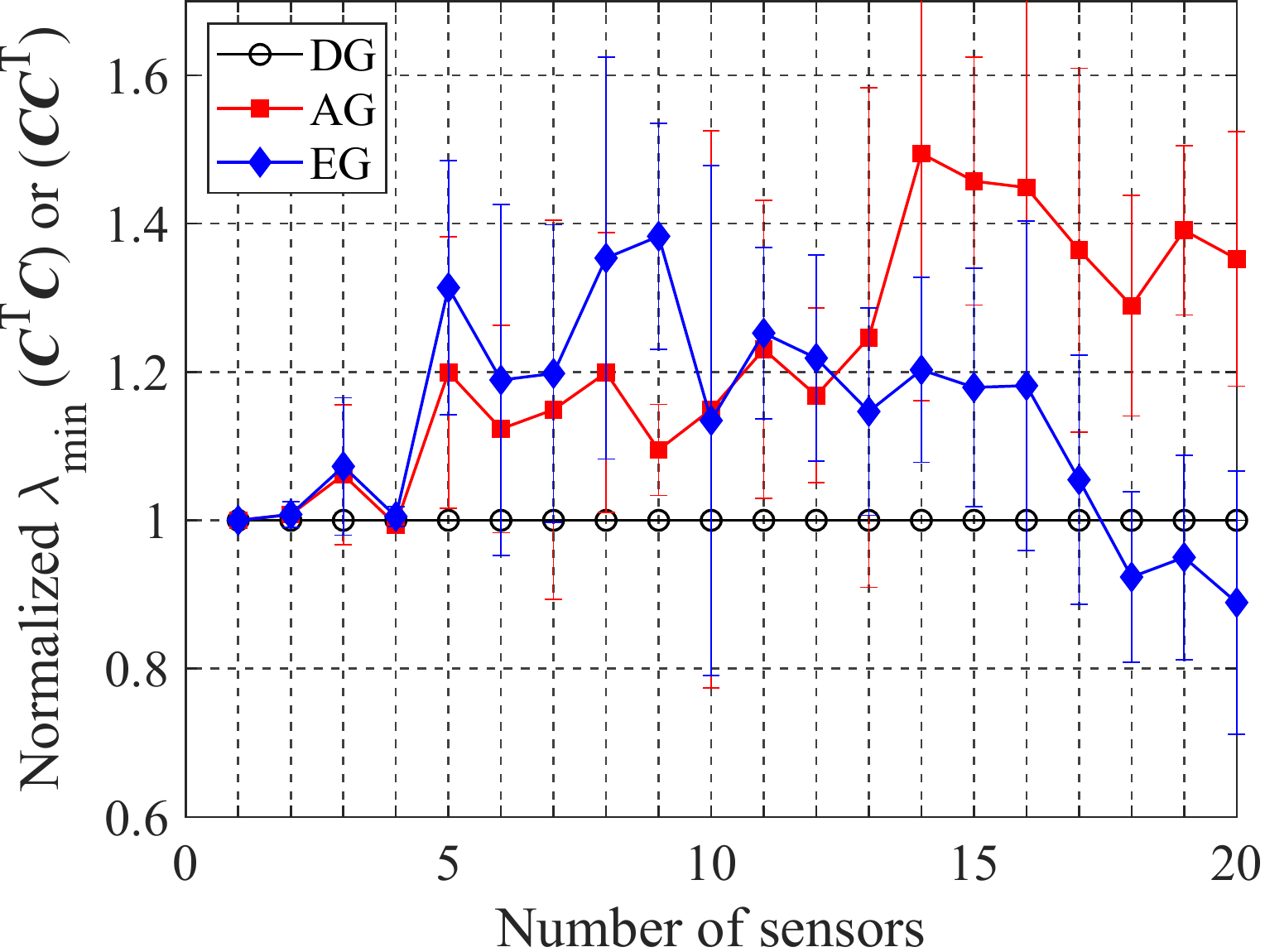}
    \caption {Normalized minimum eigenvalue of $\bm{CC}^{\top}$ ($p \le r$) or $\bm{C}^{\top}\bm{C}$ ($p > r$) against the number of sensors for for the NOAA-SST problem. The results of the random selection are omitted because they are smaller than the plotted range.}
    \label{fig:NOAA_eigenvalue}
\end{figure}

\begin{figure}[htbp]
    \centering
    \includegraphics[width=3in]{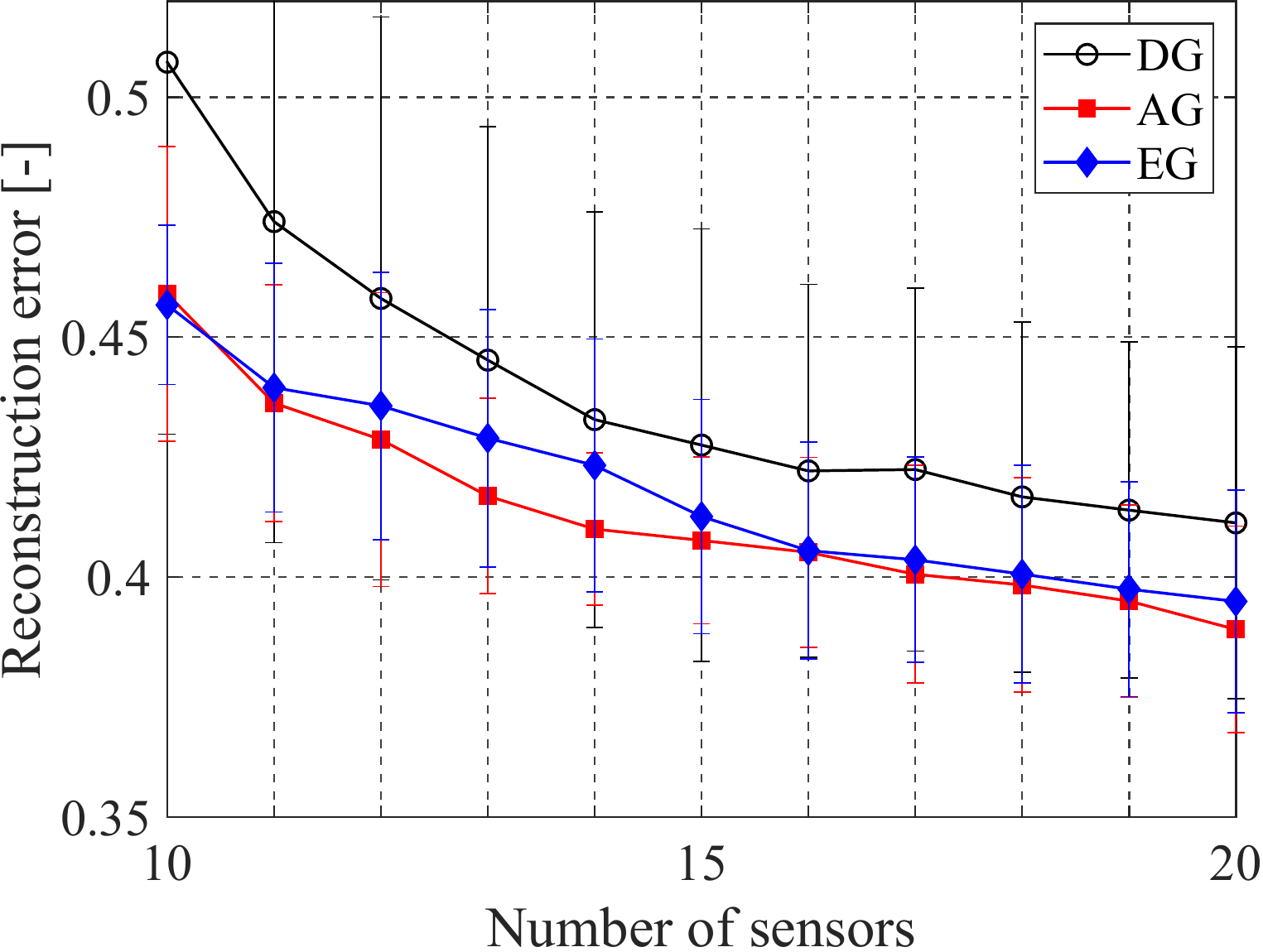}
    \caption {Reconstruction error against the number of sensors for the NOAA-SST problem. The results of the random selection are omitted because they are larger than the plotted range.}
    \label{fig:NOAA_error_CV}
\end{figure}

\begin{figure}[htbp]
    \centering
    \includegraphics[width=3in]{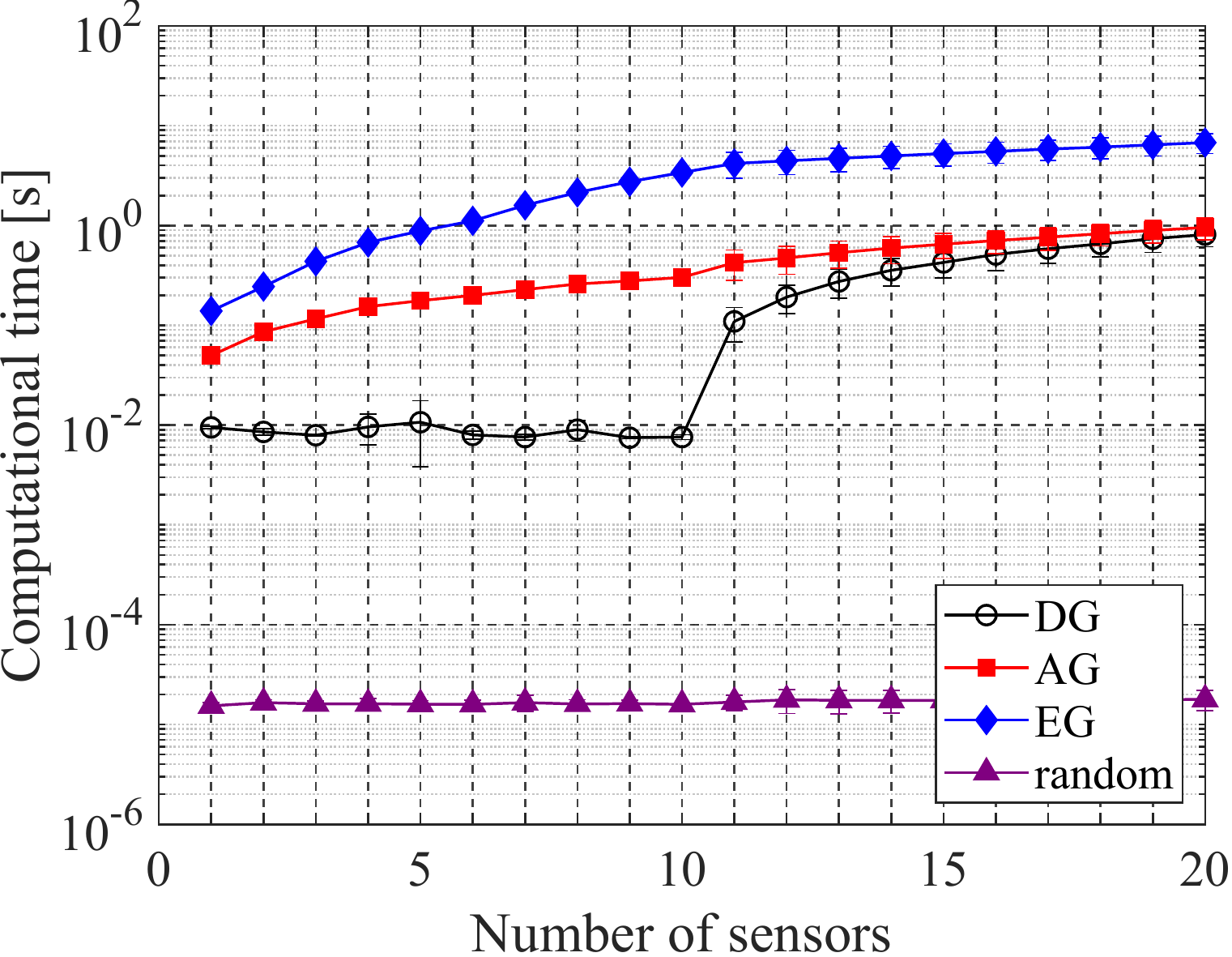}
    \caption {Computational time against the number of sensors for NOAA-SST problem}
    \label{fig:NOAA_time}
\end{figure}

\section{Conclusions}
The problem of choosing an optimal set of sensors estimating a snapshot of high-dimensional data is considered. Greedy methods for D-, A-, and E-optimality of optimal designs are considered, corresponding to designs that maximize the determinant, minimize the trace of the inverse, and maximize the minimum eigenvalue of the Fisher information matrix, respectively. 

Objective functions based on A-and E-optimality are derived, and their submodularity and monotonicity are investigated. Then, the D-, A-, and E-optimality-based greedy (DG, AG, and EG) methods are applied to two problems with different conditions, and guidance on selecting a suitable objective function for a given dataset and situation is provided. A randomly generated system with known statistics (random system problem) is used and a practical dataset of global ocean surface temperature is reconstructed with a $K$-fold cross-validation test (NOAA-SST problem). The random system problem is almost ideal in terms of statistics, while the NOAA-SST problem is nonideal since the test data statistics of noise are not exactly the same as what is assumed for the system owing to the cross-validation study. 

Trends in optimal indices concerning the sensors selected by each method do not change for the two problems. The DG and EG methods have the trends expected from the objective functions. On the other hand, it is of significance that the AG method works best not only on minimization of the trace of the inverse, but also on maximization of the minimum eigenvalue. This is because the increase in the minimum eigenvalue is strongly demanded in the objective function of the AG method while the EG method does not work significantly for maximization of the minimum eigenvalue due to the absence of submodularity.

Meanwhile, reconstruction error trends differ between the two problems. In the ideal case where the statistics of the system are well known, the determinant and the trace of the inverse are the important indices, and accordingly, the DG method is the most suitable for accurate reconstruction and low computational cost. On the other hand, in the nonideal case where the statistical assumption for the system is not valid as in the cross-validation study, the worst-case variance of the estimation error governs the performance, and accordingly, the AG method is the most suitable.
Based on the superiority of A-optimality as shown in this study, we constructed a new proximal-operator-based sensor selection algorithm based on A-optimality\cite{nagata2020data}.

In addition, these results also seem to be applicable to the sampling point of the reduced-order model (ROM) using the DEIM or Q-DEIM framework (hereafter, ROM-DEIM). This is because the procedure employed in Q-DEIM, i.e., QR decomposition, corresponds to optimization using the D-optimality criterion for the number of sensors less than or equal to that of the latent variable, as stated previously. Although oversampling in the DEIM framework is an issue, we suggest that this is straightforward to navigate using an optimal experimental design, as done in the present study. This suggests that other optimal designs, e.g., A-optimal or E-optimal designs, can also be utilized for selecting sampling points for the ROM-DEIM framework. These applications are interesting, though we will need to address whether the assumption of the ROM-DEIM framework\cite{drmac2016new,peherstorfer2020stability} is relatively ideal or not ideal in terms of the sensor selection problem and to clarify which is the better choice between D- and A-optimality. The comparison is partially demonstrated in the Appendix of reference\cite{saito2019determinant}, but is beyond the scope of this study. The applicability of the sensor selection methods in the present study to the sampling and oversampling strategy in the ROM-DEIM framework will be considered in future research.



\bibliographystyle{IEEEtran}
\bibliography{main}

\begin{IEEEbiography}
[{\includegraphics[width=1in,height=1.25in,clip,keepaspectratio]{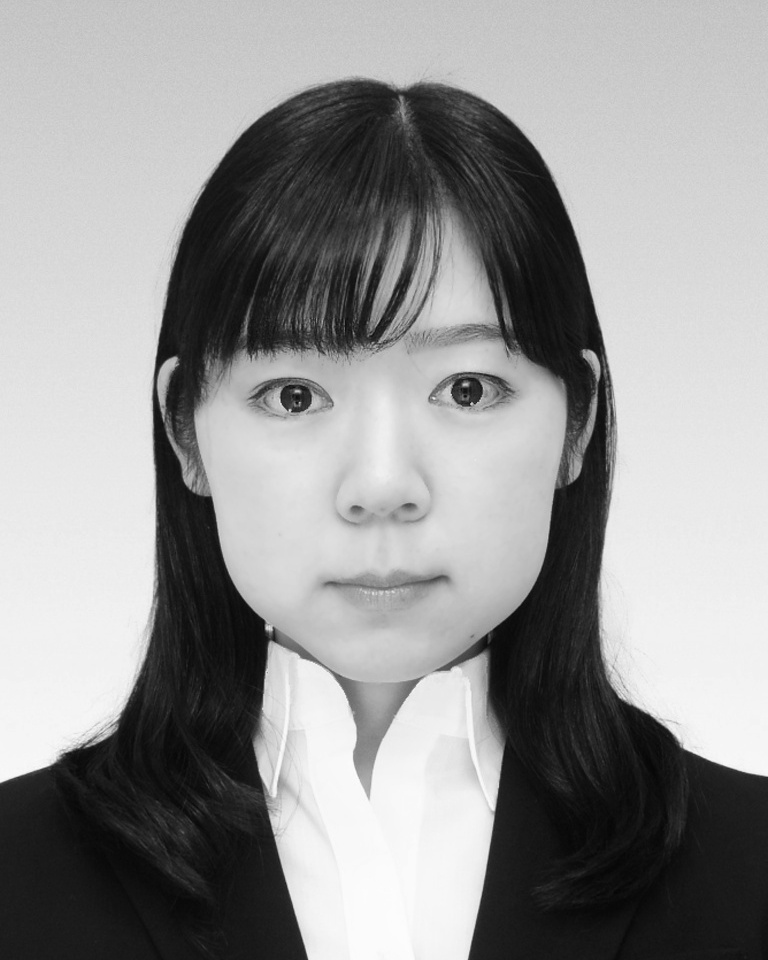}}]{Kumi Nakai} received the Ph.D. degree in mechanical systems engineering from Tokyo University of Agriculture and Technology, Japan, in 2020. From 2017-2020, she was a Research Fellow of Japan Society for the Promotion of Science (JSPS) at Tokyo University of Agriculture and Technology, Japan. Since 2020, she has been a postdoctoral researcher at Tohoku University, Japan.
\end{IEEEbiography}

\begin{IEEEbiography}
[{\includegraphics[width=1in,height=1.25in,clip,keepaspectratio]{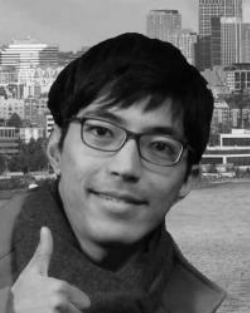}}]{Keigo Yamada} received the B.S. degree in physics from Tohoku University, Sendai, Japan, in 2019. He is a M.S. student of engineering at Tohoku University.
\end{IEEEbiography}

\begin{IEEEbiography}
[{\includegraphics[width=1in,height=1.25in,clip,keepaspectratio]{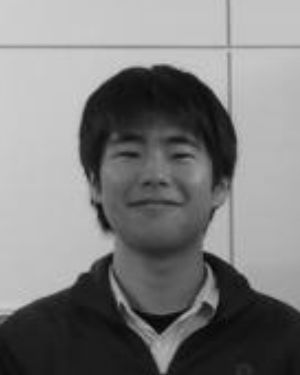}}]{Takayuki Nagata} received the Ph.D. degree in aerospace engineering from Tohoku University, Japan, in 2020. Since 2020, he has been a postdoctoral researcher at Tohoku University, Japan. 
\end{IEEEbiography}

\begin{IEEEbiography}
[{\includegraphics[width=1in,height=1.25in,clip,keepaspectratio]{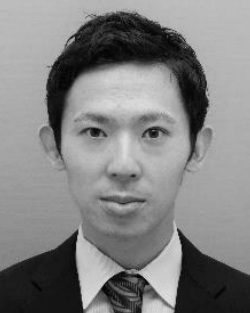}}]{Yuji Saito} received the B.S. degree in mechanical engineering, and the Ph.D. degree in mechanical space engineering from Hokkaido University, Sapporo, Japan in 2018. He is an Assistant Professor of Aerospace Engineering at Tohoku University, Sendai, Japan.
\end{IEEEbiography}

\begin{IEEEbiography}
[{\includegraphics[width=1in,height=1.25in,clip,keepaspectratio]{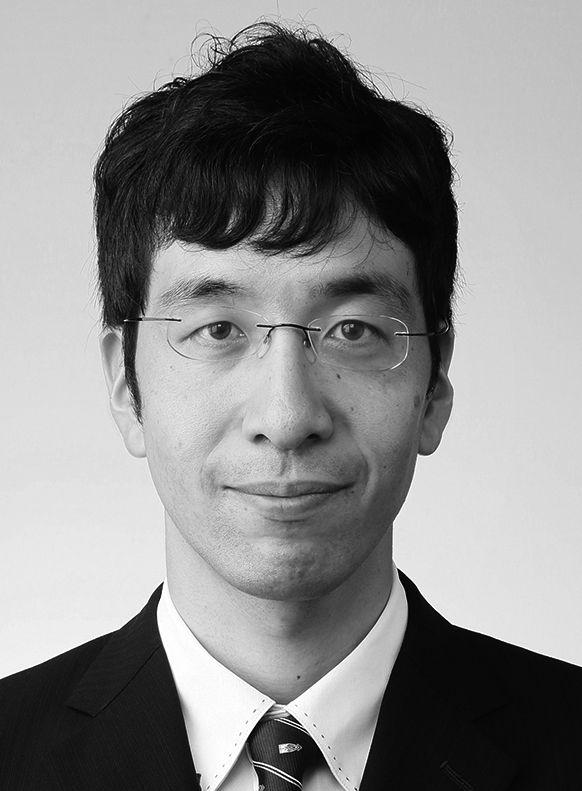}}]{Taku Nonomura} received the B.S. degree in mechanical and aerospace engineering from Nagoya University, Nagoya, Japan in 2003, and the Ph. D. degree in aerospace engineering from the University of Tokyo, Tokyo, Japan in 2008. He is currently an Associate Professor of Aerospace Engineering at Tohoku University, Sendai, Japan. 
\end{IEEEbiography}


\end{document}